\documentclass[conference]{IEEEtran}
\usepackage[boxruled]{algorithm2e}
\usepackage{cite}
\usepackage{graphicx}
\usepackage{psfrag}
\usepackage{subfigure}
\usepackage{url}
\usepackage{amsmath}
\usepackage{array}
\usepackage{amssymb}
\usepackage{amsfonts}
\usepackage{graphicx}
\usepackage{epstopdf}
\usepackage{algorithmic}

\newtheorem{lemma}{Lemma}

\newtheorem{definition}{Definition}
\newtheorem{theorem}{Theorem}
\newtheorem{conjecture}{Conjecture}
\newtheorem{example}{Example}
\newtheorem{note}{Note}
\title{Wireless Bidirectional Relaying and Latin Squares}
\begin{document}

\author{
\authorblockN{Vishnu Namboodiri, Vijayvaradharaj T. Muralidharan and B. Sundar Rajan}
\authorblockA{Dept. of ECE, IISc, Bangalore 560012, India, Email:{$\lbrace$vishnukk, tmvijay, bsrajan$\rbrace$}@ece.iisc.ernet.in
}
}

\maketitle
\thispagestyle{empty}	
\begin{abstract}
The design of modulation schemes for the physical layer network-coded two way relaying scenario is considered with the protocol which employs two phases: Multiple access (MA) Phase and Broadcast (BC) Phase. It was observed by Koike-Akino et al. that adaptively changing the network coding map used at the relay according to the channel conditions greatly reduces the impact of multiple access interference which occurs at the relay during the MA Phase and all these network coding maps should satisfy a requirement called the {\it exclusive law}. We highlight the issues associated with the scheme proposed by Koike-Akino et al. and propose a scheme which solves these issues. We show that every network coding map that satisfies the exclusive law is representable by a Latin Square and conversely, and this relationship can be used to get the network coding maps satisfying the exclusive law. Using the structural properties of the Latin Squares for a given set of parameters, the problem of finding all the required maps is reduced to finding a small set of maps for $M-$PSK constellations. This is achieved using the notions of isotopic and transposed Latin Squares. Even though, the completability of partially filled $M \times M$ Latin Square using $M$ symbols is an open problem, two specific cases where such a completion is always possible are identified and explicit construction procedures are provided. The Latin Squares constructed using the first procedure, helps towards reducing the total number of network coding maps used. The second procedure helps in the construction of certain Latin Squares for $M$-PSK signal set from the Latin squares obtained for $M/2$-PSK signal set.
\end{abstract}
\section{Background and Preliminaries}

We consider the two-way wireless relaying scenario shown in Fig.\ref{relay_channel}, where bi-directional data transfer takes place between the nodes A and B with the help of the relay R. It is assumed that all the three nodes operate in half-duplex mode, i.e., they cannot transmit and receive simultaneously in the same frequency band. The relaying protocol consists of the following two phases: the \textit{multiple access} (MA) phase, during which A and B simultaneously transmit to R and the \textit{broadcast} (BC) phase during which R transmits to A and B. Network coding is employed at R in such a way that A (B) can decode the message of B (A), given that A (B) knows its own message. 
\begin{figure}[htbp]
\centering
\subfigure[MA Phase]{
\includegraphics[totalheight=1in,width=2in]{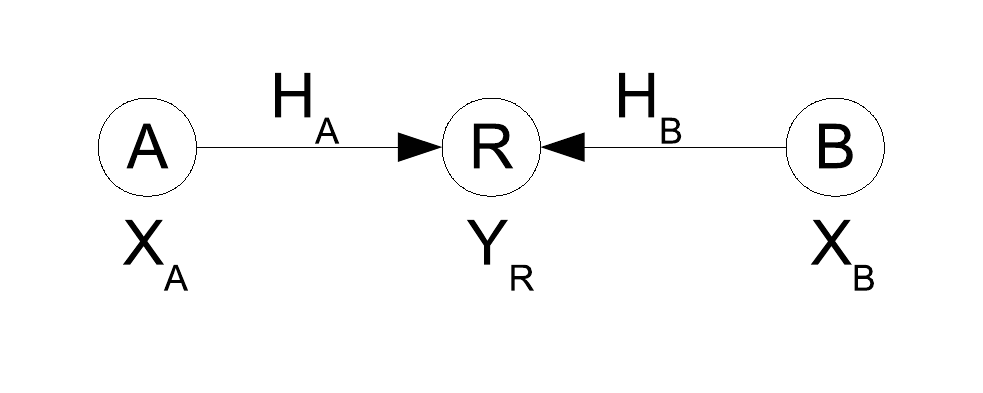}
\label{fig:phase1}
}

\subfigure[BC Phase]{
\includegraphics[totalheight=1in,width=2in]{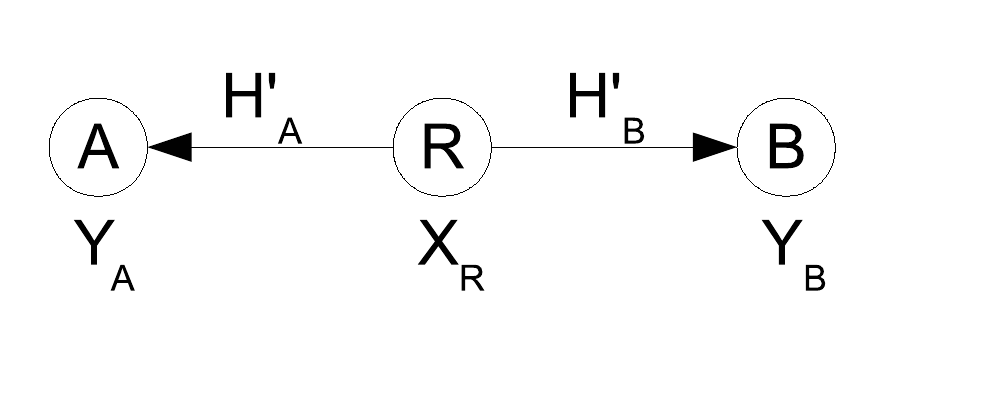}
\label{fig:phase2}
}
\caption{The Two Way Relay Channel}
\label{relay_channel}
\end{figure}
\vspace{-0.1 cm}
 
 The concept of physical layer network coding has attracted a lot of attention in recent times. The idea of physical layer network coding for the two way relay channel was first introduced in \cite{ZLL}, where the multiple access interference occurring at the relay was exploited so that the communication between the end nodes can be done using a two stage protocol. Information theoretic studies for the physical layer network coding scenario were reported in \cite{KMT},\cite{PoY}. The design principles governing the choice of modulation schemes to be used at the nodes for uncoded transmission were studied in \cite{APT1}. An extension for the case when the nodes use convolutional codes was done in \cite{APT2}. A multi-level coding scheme for the two-way relaying scenario was proposed in \cite{HeN}.

It was observed in \cite{APT1} that for uncoded transmission, the network coding map used at the relay needs to be changed adaptively according to the channel fade coefficient, in order to minimize the impact of the multiple access interference.
\subsection{Signal Model}
\subsubsection*{Multiple Access (MA) Phase}

Let $\mathcal{S}$ denote the symmetric $M$-PSK constellation used at A and B, where $M=2^\lambda$, $\lambda$ being a positive integer. Assume that A (B) wants to transmit an $\lambda$-bit binary tuple to B (A). Let $\mu: \mathbb{F}_{2^\lambda} \rightarrow \mathcal{S}$ denote the mapping from bits to complex symbols used at A and B. Let $x_A= \mu(s_A)$, $x_B=\mu(s_B)$ $\in \mathcal{S}$ denote the complex symbols transmitted by A and B respectively, where $s_A,s_B \in \mathbb{F}_{2^\lambda}$. The received signal at $R$ is given by,
\begin{align}
\nonumber
Y_R=H_{A} x_A + H_{B} x_B +Z_R,
\end{align}
where $H_A$ and $H_B$ are the fading coefficients associated with the A-R and B-R links respectively. The additive noise $Z_R$ is assumed to be $\mathcal{CN}(0,\sigma^2)$, where $\mathcal{CN}(0,\sigma^2)$ denotes the circularly symmetric complex Gaussian random variable with variance $\sigma ^2$. We assume a block fading scenario, with the ratio $ H_{B}/H_{A}$ denoted as $z=\gamma e^{j \theta}$, where $\gamma \in \mathbb{R}^+$ and $-\pi \leq \theta < \pi,$ is referred as the {\it fading state} and for simplicity, also denoted by $(\gamma, \theta).$ Also, it is assumed that $z$ is distributed according to a continuous probability distribution.   
 
 Let $\mathcal{S}_{R}(\gamma,\theta)$ denote the effective constellation at the relay during the MA Phase, i.e., 
\begin{align} 
\nonumber
 \mathcal{S}_{R}(\gamma,\theta)=\left\lbrace x_i+\gamma e^{j \theta} x_j \vert x_i,x_j \in \mathcal{S}\right \rbrace.
 \end{align}
Let $d_{min}(\gamma e^{j\theta})$ denote the minimum distance between the points in the constellation $\mathcal{S}_{R}(\gamma,\theta)$, i.e.,

{\footnotesize
\begin{align}
\label{eqn_dmin} 
d_{min}(\gamma e^{j\theta})=\hspace{-0.5 cm}\min_{\substack {{(x_A,x_B),(x'_A,x'_B)}\\{ \in \mathcal{S} \times \mathcal{S}} \\ {(x_A,x_B) \neq (x'_A,x'_B)}}}\hspace{-0.5 cm}\vert \left(x_A-x'_A\right)+\gamma e^{j \theta} \left(x_B-x'_B\right)\vert.
\end{align}
}
 
 From \eqref{eqn_dmin}, it is clear that there exists values of $\gamma e^{j \theta}$ for which $d_{min}(\gamma e^{j\theta})=0$. Let $\mathcal{H}=\lbrace \gamma e^{j\theta} \in \mathbb{C} \vert d_{min}(\gamma,\theta)=0 \rbrace$. The elements of $\mathcal{H}$ are said to be the singular fade states. An alternative definition of singular fade state is as follows:

\begin{definition}
 A fade state $\gamma e^{j \theta}$ is said to be a singular fade state, if the cardinality of the signal set $\mathcal{S}_{R}(\gamma, \theta)$ is less than $M^2$.
\end{definition}
 
 For example, consider the case when symmetric 4-PSK signal set used at the nodes A and B, i.e., $\mathcal{S}=\lbrace (\pm 1 \pm j)/\sqrt{2} \rbrace$. For $\gamma e^{j \theta}=(1+j)/2$, $d_{min}(\gamma e^{j \theta})=0$, since,
\begin{align*} 
 \left\vert \left( \dfrac{1+j}{\sqrt{2}}-\dfrac{1-j}{\sqrt{2}} \right) + \dfrac{(1+j)}{2} \left( \dfrac{-1-j}{\sqrt{2}} - \dfrac{1+j}{\sqrt{2}} \right)\right\vert=0.
 \end{align*}
\noindent 
Alternatively, when $\gamma e^{j \theta}=(1+j)/2$, the constellation $\mathcal{S}_{R}(\gamma,\theta)$ has only 12 ($<$16) points. 
 Hence $\gamma e^{j \theta}=(1+j)/2$ is a singular fade state for the case when 4-PSK signal set is used at A and B.
 
 
 
Let $(\hat{x}_A,\hat{x}_B) \in \mathcal{S} \times \mathcal{S}$ denote the Maximum Likelihood (ML) estimate of $({x}_A,{x}_B)$ at R based on the received complex number $Y_{R}$, i.e.,
 \begin{align}
 (\hat{x}_A,\hat{x}_B)=\arg\min_{({x}'_A,{x}'_B) \in \mathcal{S} \times \mathcal{S}} \vert Y_R-H_{A}{x}'_A-H_{B}{x}'_B\vert.
 \end{align}
\subsubsection*{Broadcast (BC) Phase}

Depending on the value of $\gamma e^{j \theta}$, R chooses a map $\mathcal{M}^{\gamma,\theta}:\mathcal{S} \times \mathcal{S} \rightarrow \mathcal{S}'$, where $\mathcal{S}'$ is the signal set (of size between $M$ and $M^2$) used by R during $BC$ phase. The elements in $\mathcal{S} \times \mathcal{S}$ which are mapped on to the same complex number in $\mathcal{S}'$ by the map $\mathcal{M}^{\gamma,\theta}$ are said to form a cluster. Let $\lbrace \mathcal{L}_1, \mathcal{L}_2,...,\mathcal{L}_l\rbrace$ denote the set of all such clusters. The formation of clusters is called clustering, denoted by $\mathcal{C}^{\gamma,\theta}$. 
 

The received signals at A and B during the BC phase are respectively given by,
\begin{align}
Y_A=H'_{A} X_R + Z_A,\;Y_B=H'_{B} X_R + Z_B,
\end{align}
where $X_R=\mathcal{M}^{\gamma,\theta}(\hat{x}_A,\hat{x}_B) \in \mathcal{S'}$ is the complex number transmitted by R. The fading coefficients corresponding to the R-A and R-B links are denoted by $H'_{A}$ and $H'_{B}$ respectively and the additive noises $Z_A$ and $Z_B$ are $\mathcal{CN}(0,\sigma ^2$).

In order to ensure that A (B) is able to decode B's (A's) message, the clustering $\mathcal{C}^{\gamma,\theta}$ should satisfy the exclusive law \cite{APT1}, i.e.,

{\footnotesize
\begin{align}
\left.
\begin{array}{ll}
\nonumber
\mathcal{M}^{\gamma,\theta}(x_A,x_B) \neq \mathcal{M}^{\gamma,\theta}(x'_A,x_B), \; \mathrm{where} \;x_A \neq x'_A \; \mathrm{,} \;x_B \in  \mathcal{S},\\
\nonumber
\mathcal{M}^{\gamma,\theta}(x_A,x_B) \neq \mathcal{M}^{\gamma,\theta}(x_A,x'_B), \; \mathrm{where} \;x_B \neq x'_B \; \mathrm{,} \;x_A \in \mathcal{S}.
\end {array}
\right\} \\
\label{ex_law}
\end{align}
\vspace{-.3 cm}
}

From an information theoretic perspective, the mapping $\mathcal{M}^{\gamma,\theta}$ needs to satisfy the exclusive law for the reason outlined below. Consider the ideal situation where the additive noises at the nodes are zero. It is assumed that the fading state $\gamma e^{j \theta} \notin \mathcal{H}$. The assumption is required since R can decode unambiguously to an element in $\mathcal{S} \times \mathcal{S}$ only if $\gamma e^{j \theta}$ is not a singular fade state and is justified since $\gamma e^{j \theta}$ takes values from a continuous probability distribution and the cardinality of $\mathcal{H}$ is finite. During the MA Phase, assume the relay jointly decodes correctly to the pair $(x_A,x_B)$ and transmits $X_R=\mathcal{M}^{\gamma,\theta}(x_A,x_B)$ during the BC Phase. The received complex symbols at A and B are respectively $Y_A=H'_A X_R$ and $Y_B=H'_B X_R$. At node A, the amount of uncertainty about $x_B$ which gets resolved after observing $Y_A$, $I(x_B;Y_A/x_A)=H(x_B\vert x_A)-H(x_B\vert Y_A,x_A)=H(x_B)-H(x_B \vert X_R, x_A)$ (since $x_B$ and $x_A$ are independent). Since $H(x_B \vert X_R, x_A)=0$ if and only if the mapping $\mathcal{M}^{\gamma,\theta}$ satisfies the exclusive law, the amount of uncertainty about $x_B$ ($x_A$) which gets resolved at A (B) is maximized if and only if the clustering satisfies the exclusive law. 

\begin{definition}
The cluster distance between a pair of clusters $\mathcal{L}_i$ and $\mathcal{L}_j$ is the minimum among all the distances calculated between the points $x_A+\gamma e^{j\theta} x_B ,x'_A+\gamma e^{j\theta} x'_B \in \mathcal{S}_R(\gamma,\theta)$, where $(x_A,x_B) \in \mathcal{L}_i$ and $(x'_A,x'_B) \in \mathcal{L}_j$.
\end{definition}
\begin{definition}
The \textit{minimum cluster distance} of the clustering $\mathcal{C}^{\gamma,\theta}$ is the minimum among all the cluster distances, i.e.,

{\footnotesize
\begin{align}
\nonumber
d_{min}(\mathcal{C}^{\gamma, \theta})=\hspace{-0.8 cm}\min_{\substack {{(x_A,x_B),(x'_A,x'_B)}\\{ \in \mathcal{S}\times\mathcal{S},} \\ {\mathcal{M}^{\gamma,\theta}(x_A,x_B) \neq \mathcal{M}^{\gamma,\theta}(x'_A,x'_B)}}}\hspace{-0.8 cm}\vert \left( x_A-x'_A\right)+\gamma e^{j \theta} \left(x_B-x'_B\right)\vert.
\end{align}
}

\end{definition}
The minimum cluster distance determines the performance during the MA phase of relaying. The performance during the BC phase is determined by the minimum distance of the signal set $\mathcal{S}'$. Throughout, we restrict ourselves to optimizing the performance during the MA phase. For values of $\gamma e^{j \theta}$ in the neighborhood of the singular fade states, the value of $d_{min}(\gamma e^{j\theta})$ is greatly reduced, a phenomenon referred as {\it distance shortening}. To avoid distance shortening, for each singular fade state, a clustering needs to be chosen such that the minimum cluster distance at the singular fade state is non-zero and is also maximized.  

A clustering $\mathcal{C}^{\lbrace h \rbrace}$ is said to remove a singular fade state $ h \in \mathcal{H}$, if the minimum cluster distance of the clustering $\mathcal{C}^{\lbrace h \rbrace}$ for $\gamma e^{j \theta}=h$ is greater than zero. 
Let $\mathcal{C}_{\mathcal{H}}=\left\lbrace \mathcal{C}^{\lbrace h\rbrace} : h \in \mathcal{H} \right\rbrace$ denote the set of all such clusterings. Let $d_{min}({\mathcal{C}^{\lbrace h\rbrace}},\gamma,\theta)$ denote the minimum cluster distance of the clustering $\mathcal{C}^{\lbrace h\rbrace}$ evaluated at $\gamma e^{j\theta}$. For $\gamma e^{j \theta} \notin \mathcal{H}$, the clustering $\mathcal{C}^{\gamma,\theta}$ is chosen to be $\mathcal{C}^{\lbrace h\rbrace}$, which satisfies $d_{min}({\mathcal{C}^{\lbrace h\rbrace}},\gamma,\theta) \geq d_{min}({\mathcal{C}^{\lbrace h' \rbrace}},\gamma,\theta), \forall h \neq h' \in \mathcal{H}$.

\begin{note}
The clusterings which belong to the set $\mathcal{C}_{\mathcal{H}}$ need not be distinct, since a single clustering can remove more than one singular fade state.
\end{note}

\begin{example}
In the case of BPSK, if the fade state is $\gamma=1$ and $\theta=0$ the distance between the pairs $(0,1)(1,0)$ is zero as in Fig.\ref{fig:BPSK}(a).The following clustering remove this singular fade state:
$$\{\{(0,1)(1,0)\},\{(1,1)(0,0)\}\}$$
The minimum cluster distance is non zero for this clustering.
\end{example}

 \begin{figure}[t]
\centering
\vspace{-.8 cm}
\includegraphics[totalheight=2.5in,width=2.5in]{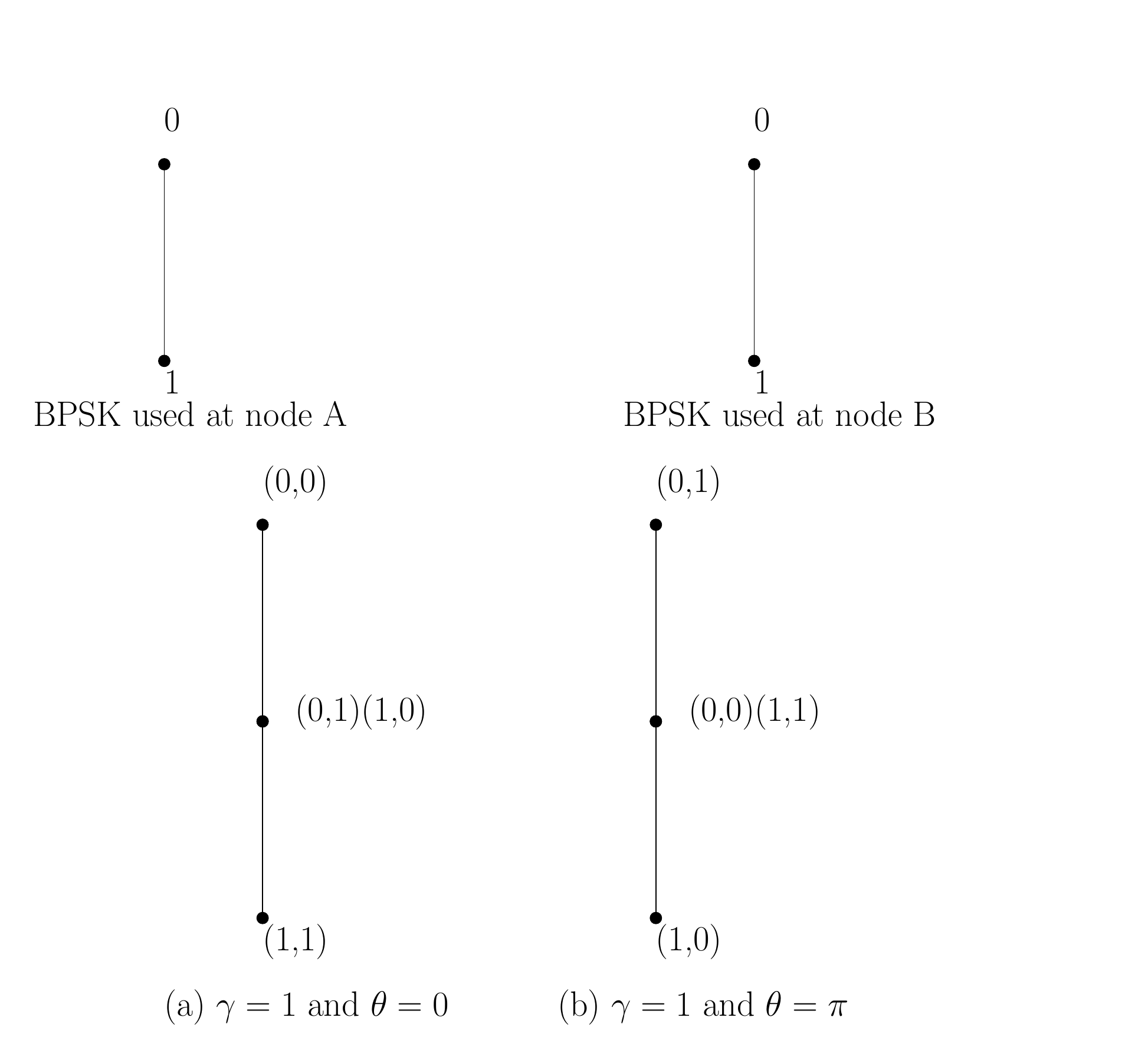}
\caption{Effective Constellation at the relay for singular fade states, when the end nodes use BPSK constellation.}     
\label{fig:BPSK}        
\end{figure}

\subsection{Issues with Koike-Akino Popovski Tarokh's approach \cite{APT1}}
It is assumed that the channel state information is not available at the transmitting nodes A and B during the MA phase. A block fading scenario is assumed and the clustering used by the relay is indicated to A and B by using overhead bits. 

The procedure suggested in \cite{APT1} to obtain the set of all clusterings, was using a computer search algorithm (closest neighbour clustering algorithm), which involved varying the fade state values over the entire complex plane, i.e., $0 \leq  \gamma < \infty$, $0 \leq \theta < 2\pi$ and finding the clustering for each value of channel realization. The total number of network codes which would result is known only after the algorithm is run for all possible realizations $\gamma e^{j \theta}$ which is uncountably infinite and hence the number of overhead bits required is not known beforehand. Moreover, performing such an exhaustive search is extremely difficult in practice, especially when the cardinality of the signal set $M$ is large.

The implementation complexity of the scheme suggested in \cite{APT1} is extremely high. It appears that, for each realization of the singular fade state, the closest neighbour clustering algorithm \cite{APT1} needs to be run at R to find the clustering. In contrast, we provide a simple criterion (Section IV A, \cite{VNR}) based on which a clustering from the set $\mathcal{C}_{\mathcal{H}}$ is chosen depending on the value of $\gamma e^{j \theta}$. In this way, the set of all values of $\gamma e^{j\theta}$ (the complex plane) is quantized into different regions, with a clustering from the set $\mathcal{C}_{\mathcal{H}}$ used in a particular region. 


In the closest neighbour clustering algorithm suggested in \cite{APT1}, the network coding map is obtained by considering the entire distance profile. The disadvantages of such an approach are two-fold. 

\begin{itemize}
\item
Considering the entire distance profile, instead of the minimum cluster distance alone which contributes dominantly to the error probability, results in an extremely large number of network coding maps. For example, for 8-PSK signal set, the closest neighbour clustering algorithm results in more than 5000 maps.
\item
 The closest neighbour clustering algorithm tries to optimize the entire distance profile, even after clustering signal points which contribute the minimum distance. As a result, for several channel conditions, the number of clusters in the clustering obtained is greater than the number of clusters in the clustering obtained by taking the minimum distance alone into consideration. This results in a degradation in performance during the BC Phase, since the relay uses a signal set with cardinality equal to the number of clusters. 
 \end{itemize}

In \cite{APT1}, to overcome the two problems mentioned above, another algorithm is proposed, in which for a given $\gamma e^{j \theta}$, an exhaustive search is performed among all the network coding maps obtained using the closest neighbour clustering algorithm and a map with minimum number of clusters is chosen. The difficulties associated with the implementation of the closest neighbour clustering algorithm carry over to the implementation of this algorithm as well.

To avoid all these problems we suggest a scheme, which is based on the removal of all the singular fade states. Since the number of singular fade states is finite (the exact number of singular fade states and their location in the complex plane are discussed in Section II), the total number of network coding maps used is upper bounded by the number of singular fade states. In fact, the total number of network coding maps required is shown to be lesser than the total number of singular fade states in Section VI. In other words, the total number of network coding maps required is known exactly, which determines the number of overhead bits required. It is shown in section III that the problem of obtaining clusterings which remove all the singular fade states reduces to completing a finite number of partially filled Latin Squares, which totally avoids the problem of performing exhaustive search for an uncountably infinite number of values.


In \cite{APT1}, an important observation was that for the case of QPSK modulation during the MA Phase, there exists several channel conditions under which the use of unconventional 5-ary signal constellation results in a larger throughput. It appears as if the mitigation of the multiple access interference, comes at the cost of degraded performance during the BC phase, because of the use of a signal set with a larger cardinality during the BC phase. By providing all the clusterings explicitly, it is shown in Section III B that the use of a signal set with larger cardinality is not required for 8-PSK signal set. In other words, the mitigation of multiple access interference does not come at the cost of degraded performance during the BC phase, for 8-PSK signal set.

The contributions and organization of the paper are as follows:
\begin{itemize}
\item It is shown that the requirement of satisfying the exclusive law is same as the clustering being represented by a Latin Square and can be used to get the clustering which removes singular fade states. In other words, it is shown that the problem of finding a clustering which removes a singular fade state reduces to filling a partially filled Latin Square. 
\item Using the properties of the set of Latin Squares for a given set of parameters, the problem of finding the set of maps corresponding to all the singular fade states can be simplified to finding the same for only for a small subset of singular fade states. Specifically, it is shown that 
\begin{enumerate}
\item For the set of all singular fade states lying on a circle, from a Latin Square corresponding to one singular fade state, Latin Squares for the other singular fade states can be obtained by appropriate permutation of the columns of the first Latin Square.
\item There is a one-to-one correspondence between a Latin Square corresponding to a  singular fade state on a circle of radius $r$ and a Latin Square corresponding to a singular fade state on a circle of radius $\frac{1}{r}.$  
\end{enumerate}
\item It is shown that the bit-wise XOR mapping can remove the singular fade state $(\gamma=1, \theta=0)$ for any  $M$-PSK, (i.e., for $M$ any power of 2) 
\item  For any $M$-PSK signal set, all the clusterings which can remove the singularities can be obtained with the aid of Latin Squares along with their isotopes. As an example, this is shown explicitly for QPSK signal set and 8PSK signal set.
\item It was shown in \cite{APT1} that for the case when QPSK signal set is used during the MA phase, there exists certain values of fade state for which a clustering with cardinality 5 needs to be used to maximize the minimum cluster distance. While the use of clusterings with cardinality 5, reduces the impact of multiple access interference, it adversely impacts the performance during the BC phase. For 8-PSK signal set, clusterings which remove the singular fade states, all of which have a cardinality of 8 are explicitly given, which implies that there is no trade-off between the MAC Phase and the BC phase.
\item
Even though, the problem of completability using $M$ symbols of an $M \times M$ Latin Square in general is unsolved, it is shown that the structure of the partially filled Latin Square in the problem considered allows the construction of explicit Latin Squares, for some singular fade states. In other words, by providing some explicit constructions, for $M$-PSK signal set, it is shown that the partially filled Latin Squares corresponding to some singular fade states can be completed using $M$ symbols itself.
\item
It was observed in \cite{APT1} that for 4 PSK signal set there exists clusterings which remove more than one singular fade state. For any $M$-PSK signal set, certain special cases when multiple singular fade states are removed by the same Latin Square are identified. A construction algorithm is provided to obtain such Latin Squares. Each one of the Latin Squares obtained using the algorithm provided removes $M^2/8$ singular fade states. 
\item
An explicit construction procedure is provided to show that certain Latin Squares for $M$-PSK signal set, are obtainable from the Latin Squares for $M/2$-PSK signal set. 
\item
When the end nodes use constellations of different sizes $M(=2^\lambda)$ and $N(=2^\mu)$, to get the clusterings it is required to fill Latin Rectangles. These Latin rectangles are obtained by removing appropriate columns of the Latin Squares constructed for the case when both the end nodes use constellations of same size which is $\max(M,N)$.
\end{itemize}

\section{SINGULAR FADE STATES FOR M-PSK SIGNAL SET}

Throughout the paper the points in the symmetric $M$-PSK signal set are assumed to be of the form $e^{j (2k+1) \pi/M},0 \leq k \leq M-1$ and $M$ is of the form $2^\lambda$, where $\lambda$ is a positive integer.
Let $\Delta\mathcal{S}$ denote the difference constellation of the $M$-PSK signal set $\mathcal{S}$, i.e., $\Delta\mathcal{S}=\lbrace s_i-s'_i \vert  s_i, s'_i \in \mathcal{S}\rbrace$.

For any M-PSK signal set, the set $\Delta\mathcal{S}$ is of the form,

{\footnotesize
\begin{align}
\nonumber
\Delta\mathcal{S}=&\left\lbrace 0\right\rbrace\cup \left\lbrace 2\sin(\pi n /M) e^{j k 2 \pi/M}  \vert{n \; \textrm{odd} }\right\rbrace\\
\nonumber
&\hspace{2 cm}\cup\left\lbrace 2\sin(\pi n /M) e^{j (k 2 \pi/M +  \pi/M)}\vert{n \; \textrm{even} }\right\rbrace,
\end{align}
}where $1 \leq n \leq M/2$ and $0 \leq k \leq M-1$.

In other words, the non-zero points in $\Delta\mathcal{S}$ lie on $M/2$ circles of radius $2\sin(\pi n/M), 1 \leq n \leq M/2$ with each circle containing $M$ points. The phase angles of the $M$ points on each circle is $2 k \pi/M$, if $n$ is odd and $2k \pi/M+\pi/M$ if $n$ is even, where $0 \leq k \leq M-1$. For example the difference constellation for 4-PSK and 8-PSK signal sets are shown in Fig. \ref{4psk_diff} and Fig. \ref{8psk_diff} respectively.

\begin{figure}[htbp]
\centering
\includegraphics[totalheight=3.5in,width=6in]{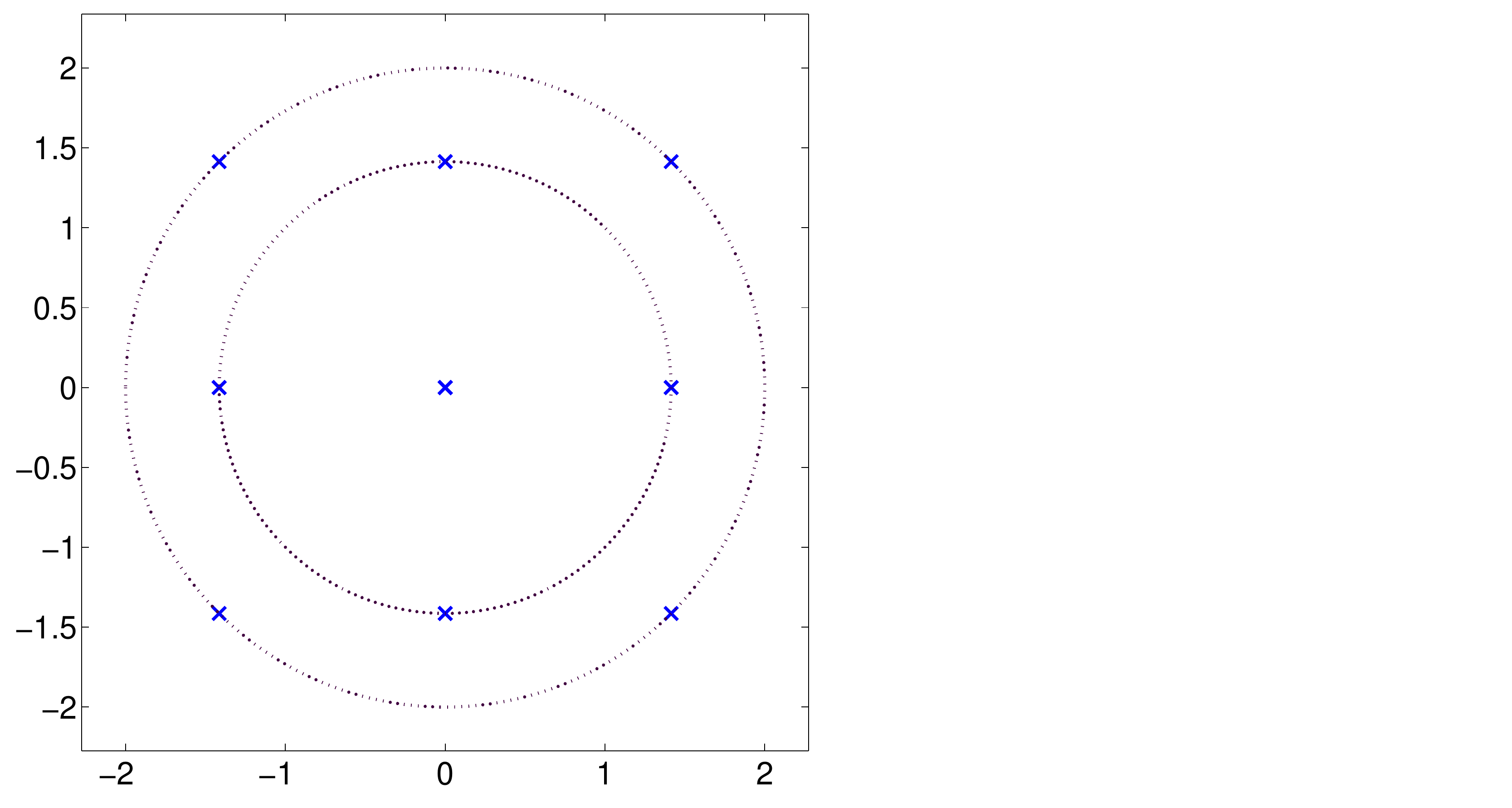}
\caption{Difference constellation for 4-PSK signal set}	
\label{4psk_diff}	
\end{figure}

\begin{figure}[htbp]
\centering
\includegraphics[totalheight=3in,width=5in]{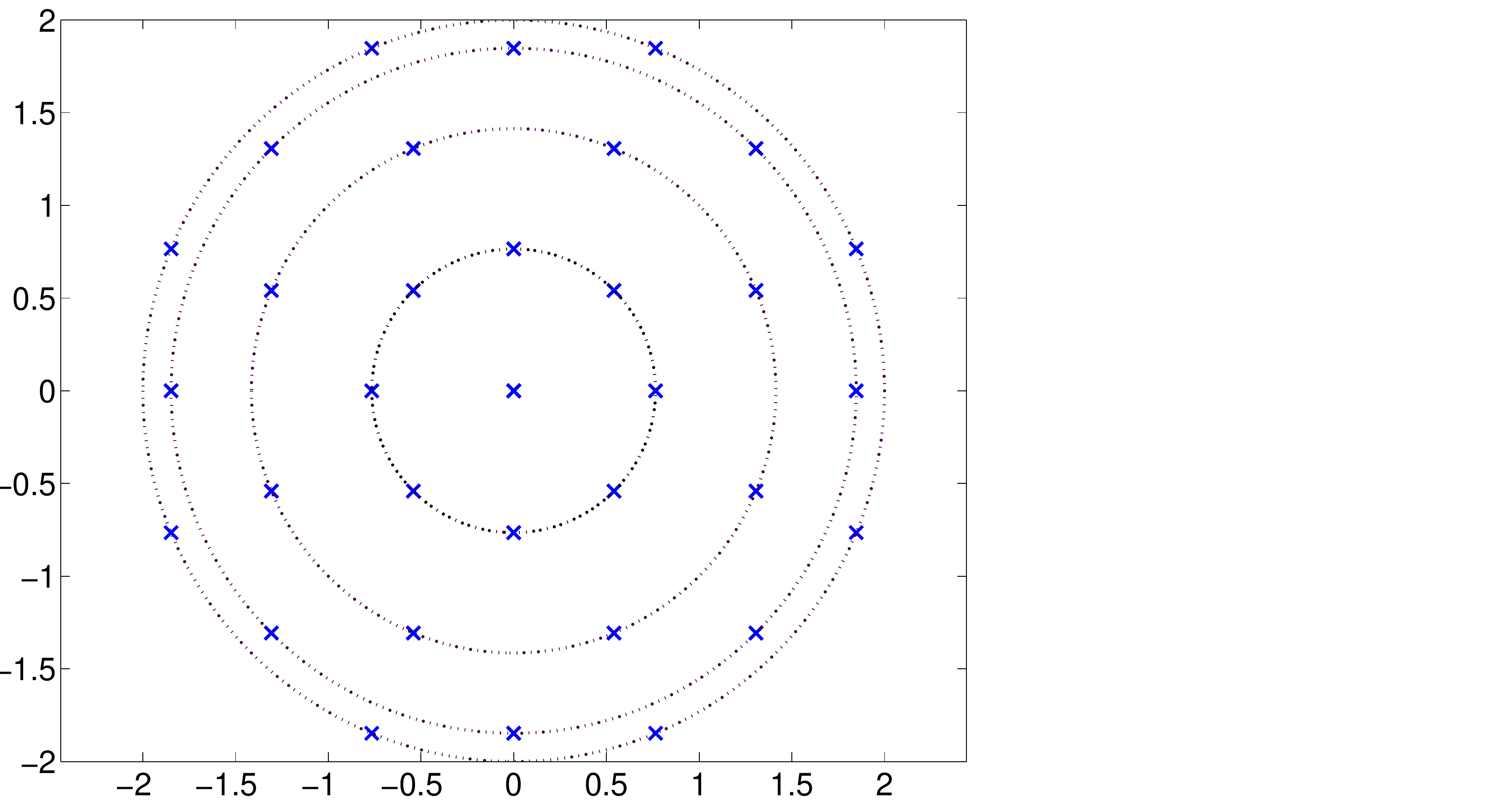}
\caption{Difference constellation for 8-PSK signal set}	
\label{8psk_diff}	
\end{figure}

Let us define,

\begin{align}
\nonumber
x_{k,n}=\left\lbrace
\begin{array}{ll}
\nonumber
2\sin(\pi n /M) e^{j k 2 \pi/M} {\;\textrm{if}\; n \; \textrm{is} \; \textrm{odd}, }\\
2\sin(\pi n /M) e^{j (k 2 \pi/M +  \pi/M)} {\;\textrm{if}\; n \; \textrm{is} \; \textrm{even}, }
\end{array}
\right.\\
\label{eqn_diff}
\end{align}
where $1 \leq n \leq M/2$ and $0 \leq k \leq M-1$.

From \eqref{eqn_dmin}, it follows that the singular fade states are of the form, 

{\footnotesize
\vspace{-0.2 cm}
\begin{align*}
\gamma_{s} e^{j \theta_{s}}=-x_{k,n}/x_{k',n'}, 
\end{align*}
}for some $x_{k,n},x_{k',n'} \in \Delta\mathcal{S}$.

\begin{lemma}
\label{lemma_trig}
For integers $k_1$, $k_2$, $l_1$ and $l_2$, where $$1 \leq k_1,k_2,l_1,l_2 \leq \frac{M}{2},k_1 \neq k_2  \textrm{ and }  l_1 \neq l_2,$$

{\footnotesize
\vspace{-0.4 cm}
\begin{align}
\nonumber
\dfrac{\sin(k_1 \pi/M)}{\sin(k_2 \pi/M)}=\dfrac{\sin(l_1 \pi/M)}{\sin(l_2 \pi/M)},
\end{align}
}if and only if $k_1 = l_1$ and  $k_2 = l_2$.
\begin{proof}
See \cite{VNR}.
\end{proof}
\end{lemma}

The following lemma  gives the location of the singular fade states in the complex plane.
\begin{lemma}
\label{lemma_singularity}
The singular fade states other than zero lie on $M^2/4-M/2+1$ circles with $M$ points on each circle, with the radii of the circles given by $\sin(k_1\pi/M)/\sin(k_2 \pi/M)$, where $1 \leq k_1,k_2 \leq M/2$. The phase angles of the $M$ points on each one of the circles are given by $k2\pi/M$, $0 \leq k \leq M-1$, if both $k_1$ and $k_2$ are odd or both are even and $k2\pi/M+\pi/M$, $0 \leq k \leq M-1$, if only one among $k_1$ and $k_2$ is odd.
%
\end{lemma}

From Lemma \ref{lemma_singularity}, it follows that if $\gamma e^{j\theta}$ is a singular fade state, $\frac{1}{\gamma} e^{-j \theta}$ is also a singular fade state.
\begin{example}
For the case when 4-PSK signal set is used during the MA Phase, the singular fade states lie on three circles as shown in Fig. \ref{4psk_sing}.
\end{example}
\begin{example}
For the case when 8-PSK signal set is used during the MA Phase, the singular fade states lie on thirteen circles as shown in Fig. \ref{8psk_sing}.
\end{example}

\begin{figure}[htbp]
\centering
\includegraphics[totalheight=3in,width=5.5in]{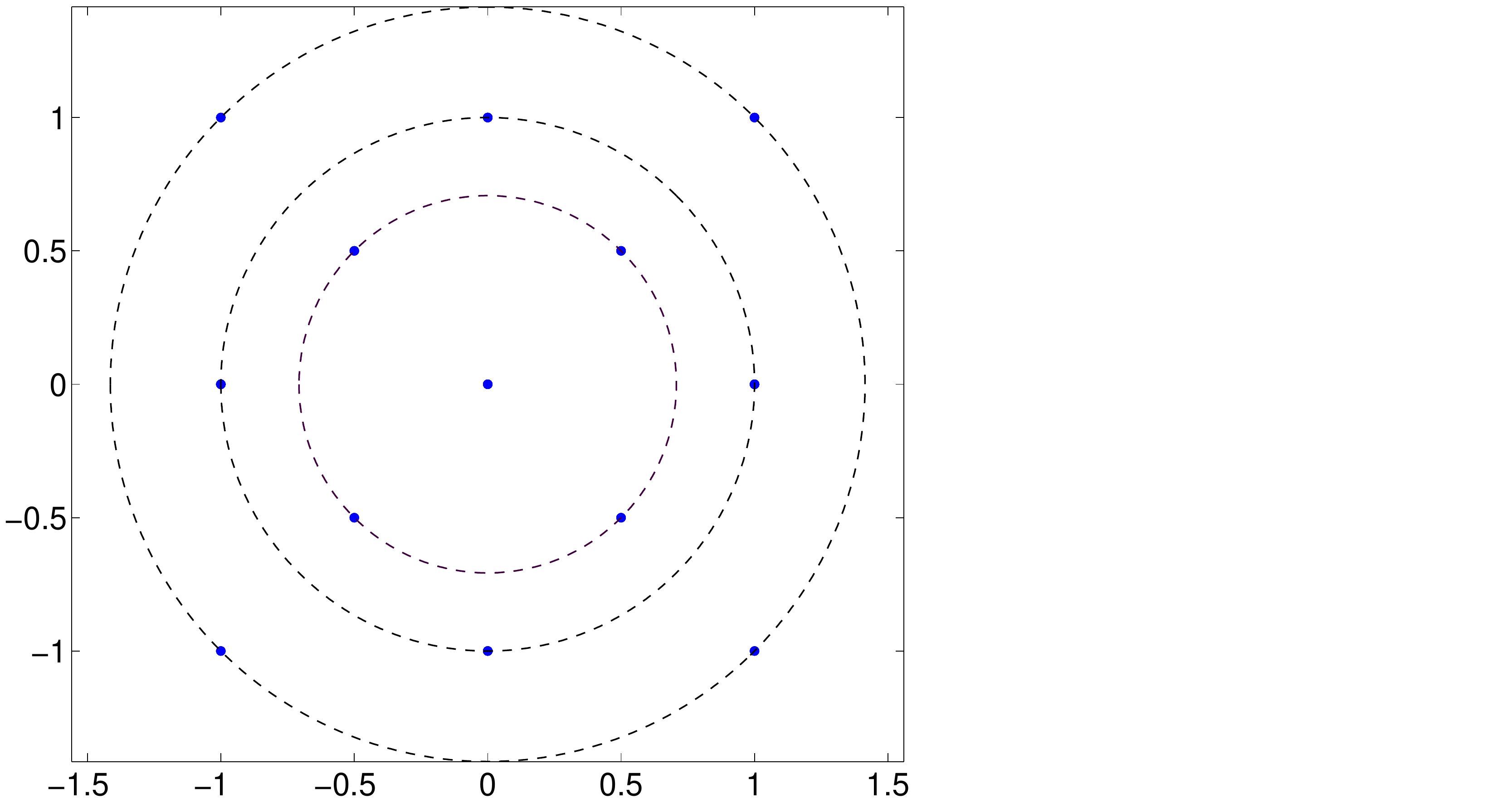}
\caption{Singular fade states for 4-PSK signal set}	
\label{4psk_sing}	
\end{figure}

\begin{figure}[htbp]
\centering
\includegraphics[totalheight=2.5in,width=4.5in]{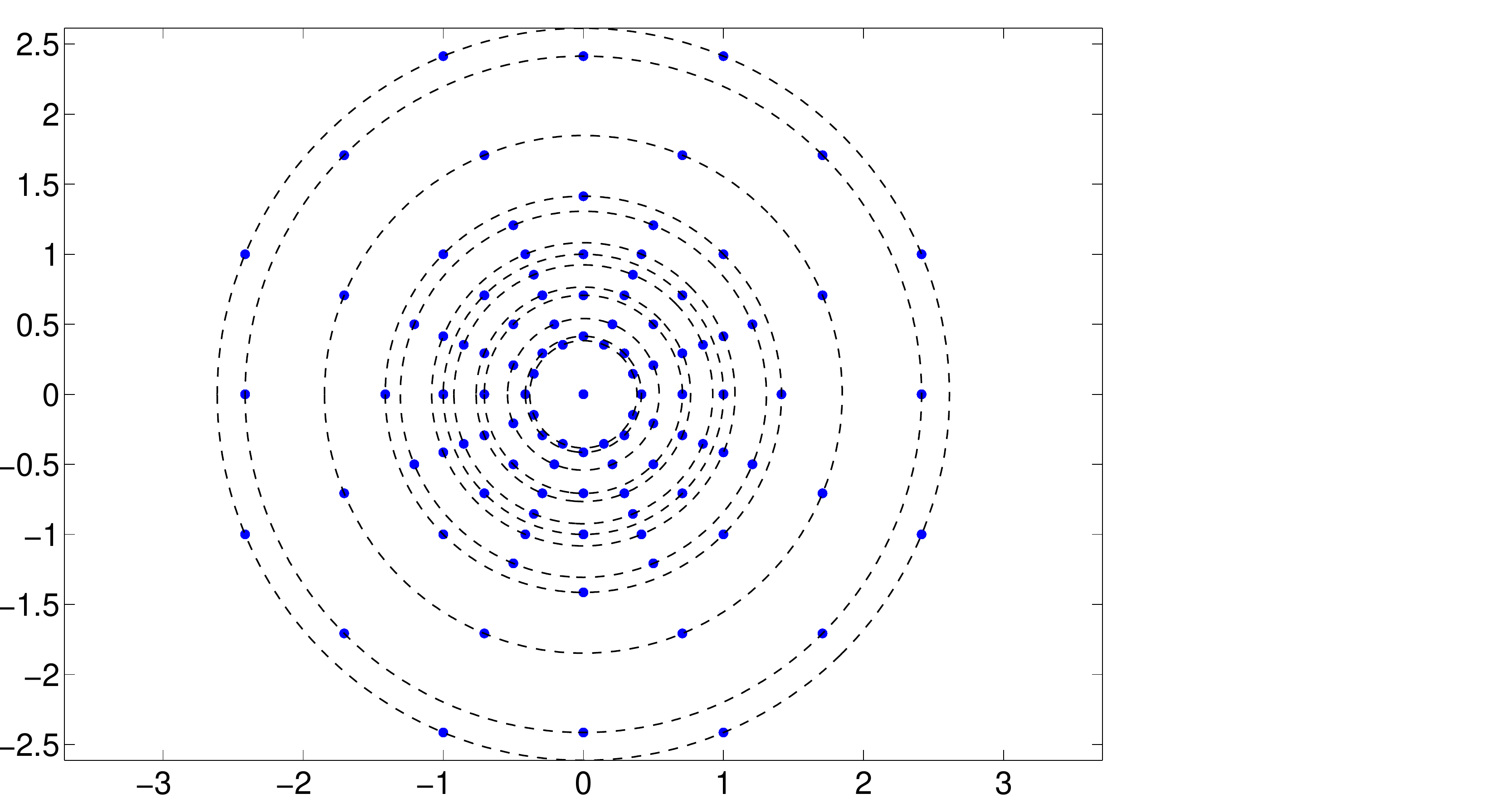}
\caption{Singular fade states for 8-PSK signal set}	
\label{8psk_sing}	
\end{figure}

\section{The Exclusive Law and Latin Squares}

\textit{Definition 4:} \cite{Rod} A Latin Square L of order $M$ on the symbols from the set $\mathbb{Z}_t=\{0,1, \cdots ,t-1\}$ is an \textit{M} $\times$ \textit{M}  array, in which each cell contains one symbol and each symbol occurs at most once in each row and column. 

Let the nodes A and B use the same constellation of size $M.$ Consider an \textit{M} $\times$ \textit{M} array at the relay with the rows (columns) indexed by the constellation point used by node A (B), i.e., symbols from the set \{{0, 1, 2,}\dotso, {$M$-1}\}. The relay is allowed to use any constellation with size $\textit{t}\geq\text{M}$ (using $t > M$ may lead to some advantages, see \cite{APT1}).    Our aim is to cluster the $M^{2}$ slots in the $\textit{M}\times\textit{M}$ array such that the exclusive law is satisfied. To do so, we will fill in the slots in the array with the elements of set $\mathbb{Z}_{t}$ in such a way that \eqref{ex_law} is satisfied, and the clusters are obtained by taking all the row-column pairs $(i,j),$ $i,j \in {0,1,\cdots,M-1},$ such that the entry in the $(i,j)-$th slot is the same symbol from  $\mathbb{Z}_{t}.$ The specific symbols from $\mathbb{Z}_t$ are not important, but it is the set of clusters that are important. 
Now, it is easy to see that if the exclusive law need to be satisfied, then the clustering should be such that an element in a row and also in a column cannot be repeated in the same row and column. Thus all the relay clusterings which satisfy the exclusive law form Latin Squares. Hence, we have the following:

\textit{ All the relay clusterings which satisfy the exclusive law forms Latin Squares, when the end nodes use constellations of same size.}

With this observation, the study  of clustering which satisfies the exclusive law can be equivalently carried out as the study  of  Latin Squares with appropriate parameters.

\subsection{Removing Singular fade states and Constrained Latin Squares}

The relay can manage with constellations of size $M$ in BC phase, but it is observed that in some cases relay may not be able to remove the singular fade states with $t=M$ and results in severe performance degradation in the MA phase \cite{APT1}. Let $(k,l)(k^{\prime},l^{\prime})$ be the pairs which give same point in the effective constellation $\mathcal{S}_R$ at the relay for a singular fade state, where $k,k^{\prime},l,l^{\prime} \in \{0,1,....,M-1\}$. Let $k,k^{\prime}$ be the constellation points used by node A and $l,l^{\prime}$ be the constellation points used by node B. If they are not clustered together, the minimum cluster distance will be zero. To avoid this, those pairs should be in same cluster. This requirement is termed as {\it singularity-removal constraint}. So, we need to obtain Latin Squares which can remove singular fade states and with minimum value for $t.$ Therefore, initially we will fill the slots in the $\textit{M}\times\textit{M}$ array  such that for the slots corresponding to a singularity-removal constraint the same element will be used to fill slots. This removes that particular singular fade state. Such a partially filled Latin Square is called a {\it Constrained Partial Latin Square}. After this,  to make this a Latin Square, we will try to fill the other slots of the  partially filled, Constrained Partial Latin Square with minimum number of symbols from the set $\mathbb{Z}_{t}.$
\begin{example}
Consider the case where both A and B uses BPSK as shown in Fig. \ref{fig:BPSK}, where the effect of noise is neglected. There are two singular fade states, one at $(\gamma =1, \theta=0)$ and the other at $(\gamma =1, \theta =\pi).$  We try to eliminate the singular fade states one by one. First to remove $(\gamma=1, \theta=0),$  the symbol in $0^{th}$ row $1^{st}$ column (henceforth the slot (0,1)) and symbol in (1,0) should be same.  Otherwise, the  minimum cluster distance will be zero. We are using symbol 1 (choice of this symbol will not alter the clustering) for this and we will get the Constrained Partial Latin Square as in Table. \ref{LL}.  This uniquely completes to the Latin Square in Table. \ref{LLL}. Notice that this will remove the singular fade state  $(\gamma = 1, \theta = \pi)$ also.  This Latin Square corresponds to the bit-wise XOR mapping,  but with higher order constellations the number of singular fade states increases and bit-wise XOR cannot remove (will be seen in the sequel) all the singular fade states. 
\end{example}
\begin{table}
\parbox{.45\linewidth}{
\centering
 \begin{tabular}{|c|c|}
\hline  & 1 \\ 
\hline 1 &  \\ 
\hline 
\end{tabular} 
\caption{Partially filled Latin Square}
\label{LL}
}
\hfill
\parbox{.45\linewidth}{
\centering
\begin{tabular}{|c|c|}
\hline 0 & 1 \\ 
\hline 1 & 0 \\ 
\hline
\end{tabular} 
\caption{Completely Filled Latin Square}
\label{LLL}
}
\end{table} 
Let the $M-$PSK points be $e^{j(2k+1) \pi/M}$, where $0 \leq k \leq M-1$. For simplicity, by the point $k$ we mean the point $e^{j(2k+1) \pi/M}.$  The following lemma shows that bit-wise XOR mapping removes the singular fade state $(\gamma=1,\theta=0)$, for any $M$-PSK signal set. 
%
\begin{lemma}
When the user nodes  use $2^{\lambda}$-PSK constellations, the singular fade state $(\gamma=1, \theta=0)$ is removed by bit-wise XOR mapping (denoted by $\oplus$), for all $\lambda.$
\end{lemma}
\begin{proof}
For $\lbrace (k,k^{\prime}), (l,l^{\prime})\rbrace$ to be a singularity constraint corresponding to the singular fade state $(\gamma=1, \theta=0)$,

{\footnotesize
\vspace{-.3 cm}
\begin{align}
\nonumber
&-\frac{e^{ \frac{jk\pi}{M}}-e^{\frac{jk^{\prime}\pi}{M}}}{e^{\frac{jl\pi}{M}}-e^{\frac{jl^{\prime}\pi}{M}}}=1,\textrm{ i.e.,}\\
\label{eqn_xor}
&\frac{\sin\left[\pi(k-k^{\prime})/2^{\lambda}\right]} {\sin \left[\pi(l^{\prime}-l)/2^{\lambda}\right]}e^{j \frac{\pi}{M}(k+k;-l-l')}=1.
\end{align}
}

Equating the amplitudes on both sides of \eqref{eqn_xor}, one gets $$ \sin\left[\pi(k-k^{\prime})/2^{\lambda}\right] = \sin \left[\pi(l^{\prime}-l)/2^{\lambda}\right], $$ 
which leads to  the following two cases:

\noindent 
{\it Case (i):} $k-k^{\prime}= l^{\prime}-l$ \\
{\it Cases (ii):}$\dfrac{\pi(k-k^{\prime})}{2^{\lambda}} = \pi-\dfrac{\pi(l^{\prime}-l)}{2^{\lambda}} \implies k-l = k^{\prime}-l^{\prime}+2^{\lambda}.$
Equating the phases on both the sides of \eqref{eqn_xor} leads to 
\begin{equation}
\label{thetaresult}
\frac{\pi}{2^{\lambda}}(k+k^{\prime}-l-l^{\prime})=0 \implies k+k^{\prime} = l+l^{\prime}.
\end{equation} 
Combining {\it Case (i)} above and \eqref{thetaresult} gives $(k^{\prime},l^{\prime}) = (l,k),$
i.e.,  the singularity-removal constraint is of the form $\{(k,l)(l,k)\}$. In other words, the clustering should satisfy this symmetry.

Combining {\it Case (ii)} above and \eqref{thetaresult} leads to $k= l+2^{\lambda-1}$ irrespective of $k^{\prime},l^{\prime}$ and $k^{\prime} = l^{\prime}+2^{\lambda-1}$ irrespective of $k,l.$ In other words,  $\{(l+2^{\lambda-1},l)\}$, $l \in \{0,1,....,2^{\lambda}\}$ is the set of singularity-removal constraints.

From the above, one can conclude that a clustering which removes the singular fade state ($\gamma=1, \theta=0$) should have \\
(i) A symmetric Latin Square, meaning that the cells $(k,l)$ and $(l,k)$ should have the same symbol.\\
(ii) A Latin Square with the symbols in the cells  $\{(l+2^{\lambda-1},l)\},$ and  $l \in \{0,1,....,2^{\lambda}\}$ being the same.

The Latin Square produced by bit-wise XOR mapping is clearly symmetric. Moreover, the quantity  $(l+2^{\lambda-1})\oplus l$ is always equal to $2^{\lambda-1}$ for all values of $l,$ i.e., the symbols in all the cells of the set   $\{(l+2^{\lambda-1},l)\}, ~~ l \in \{0,1,....,2^{\lambda}\}$ are the same. Hence the  XOR map removes the singular fade state ($\gamma=1, \theta=0$).
\end{proof}
\textit{Definition 5:} \cite{Sto} Two Latin Squares $L$ and $L$ $^{\prime}$ (using the same symbol set) are isotopic if there is a triple $(\textit{f,g,h}),$ where $f$ is a row permutation, $g$ is a column permutation and $h$ is a symbol permutation, such that applying these permutations on  $L$ gives $L^{\prime}.$
\begin{lemma}
\label{lemma_col_perm}
Two Latin Squares $L$ and $L^\prime$ which remove the singular fade states $(\gamma, \theta)$ and $(\gamma, \theta^{\prime})$, respectively, (i.e., two singular fade states on the same circle), are Isotopic that are obtainable one from another by a column permutation alone. If $\theta'-\theta = k\frac{2\pi}{M}$, $L'$ can be obtained by cyclic shifting of the columns of $L$, $k$ times in the anticlockwise direction. 
\end{lemma}
\begin{proof}
Let $L$ and $L^\prime$, respectively remove the singular fade states $(\gamma, \theta)$ and $(\gamma, \theta^\prime).$ 

The effect of rotation in the $z-$plane by an angle $\theta^\prime - \theta$ due to channel fade coefficients $H_A$ and $H_B$ can be viewed equivalently as a relative rotation of the constellation used by B by an angle $\theta^\prime - \theta$ with respect to the constellation used by A and no relative rotation between the channel fade coefficients $H_A$ and $H_B.$ Let $S$ and $S^{\prime}$ be the resulting rotated constellations after rotation in the constellation of node B corresponding to an angle $\theta^\prime - \theta.$ 

Since there are $M$ singular fade states for a specific $\gamma$, (shown in \cite{VNR}), and they are all spaced by same angular separation,  $\theta^\prime - \theta$ is an integer multiple of $2\pi/M$ which is an  angular separation of the $M$-PSK constellation points. That is, a rotation in the channel by an angle $\theta^\prime - \theta$  is equivalent to a rotation in the constellation points in the $M$-PSK constellation. So, we can obtain the Latin Square L$^{\prime}$ by column permutations in L, since the columns are indexed by constellation points used by node B. This means, if we obtain the Latin Square for a singular fade state $(\gamma, \theta),$ then by appropriately shifting the columns we obtain the Latin Squares that remove all the other singular fade states of the form $(\gamma, \theta^\prime).$ This completes the proof.
\end{proof}

From \ref{lemma_col_perm}, it follows that for each circle, it is enough if we obtain one Latin Square which removes a singular fade state on that circle. The Latin Squares which remove the other singular fade states can be obtained by column permutation. For example, all the Latin Squares which remove the singular fade states on the unit circle can be obtained from the bit-wise XOR map by column permutation.

\textit{Definition 6:} A Latin Square $L^T$ is said to be the Transpose of a Latin Square $L$, if $L^T(i,j)=L(j,i)$ for all $i,j \in \{0,1,2,..,M-1\}.$

Recall from Section II that if $\gamma e^{j\theta}$ is a singular fade state, then $\frac{1}{\gamma}e^{-j \theta}$  is also a singular fade state. The following Lemma shows that the transpose of the Latin Square which removes $\gamma e^{j\theta}$ removes the singular fade state $\frac{1}{\gamma}e^{-j \theta}$.
\begin{lemma}
\label{lemma_trans}
If the Latin Square $L$ removes the singular fade state $(\gamma, \theta),$ then the  Latin Square $L^T$ removes the singular fade state $(\frac{1}{\gamma}, -\theta).$
\end{lemma}
\begin{proof}
Let $\{(k_1,l_1)(k_2,l_2)\}$ be a singularity-removal constraint for the singular fade state $(\gamma, \theta).$  Then,\\ 
$$\gamma=\dfrac{\sin\left(\frac{\pi(k_{1}-k_{2})}{M}\right)}{\sin\left(\frac{\pi(l_{2}-l_{1})}{M}\right)} \mbox{  and  }\theta=\dfrac{\pi}{M}(k_{1}+k_{2}-l_{1}-l_{2}).$$
Taking transpose in the constraint we will obtain $\{(l_1,k_1)(l_2,k_2\}.$ Let this constraint correspond to the singular fade state $(\gamma^\prime,  \theta^\prime).$  Then,

$$\gamma^\prime =\frac{\sin\left(\frac{\pi(l_{1}-l_{2})}{M}\right)}{\sin\left(\frac{\pi(k_{2}-k_{1})}{M}\right)} =\dfrac{\sin\left(\frac{\pi(l_{2}-l_{1})}{M}\right)}{\sin\left(\frac{\pi(k_{1}-k_{2})}{M}\right)}=1/\gamma.$$
\noindent
Similarly,
$$\theta^\prime=\frac{\pi}{M}(l_{1}+l_{2}-k_{1}-k_{2}) =-\frac{\pi}{M}(k_{1}+k_{2}-l_{1}-l_{2}) =-\theta.$$
This completes the proof.
%
\end{proof}

 Lemma \ref{lemma_trans} implies that it is enough to obtain the Latin Squares which remove those singular fade states which lie on or inside the unit circle centered at the origin. The Latin Squares which remove the singular fade states which lie outside the unit circle can be obtained by taking transpose.

\section{Illustrations}

\subsection{End nodes use QPSK}
There are 12 singular fade states (shown in Fig.\ref{4psk_sing}), when both the end nodes A and B use QPSK as in Fig.\ref{fig:nodes}.


\begin{figure}[t]
\centering
\vspace{-2 cm}
\includegraphics[totalheight=4in,width=3in]{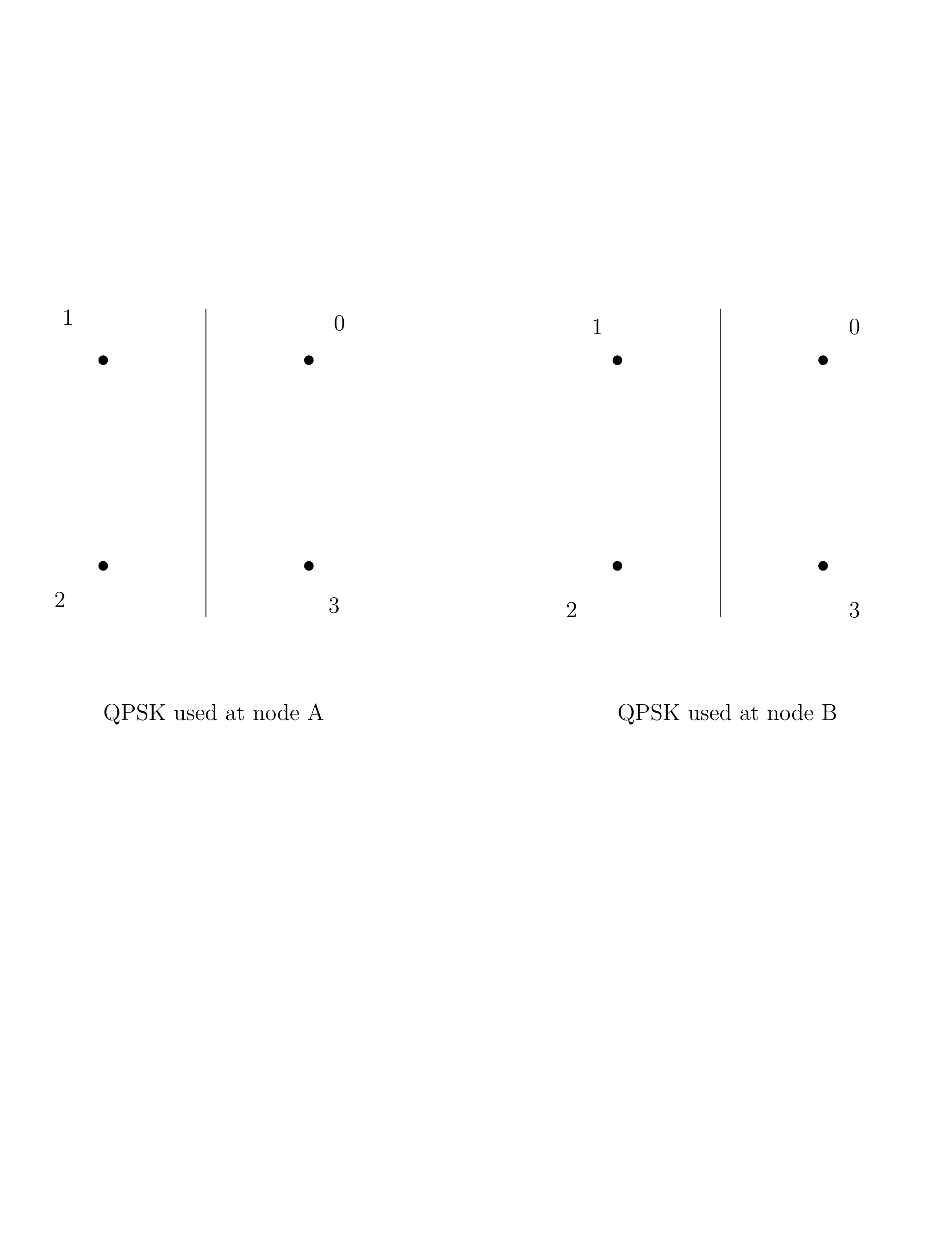}
\vspace{-4cm}
\caption{QPSK Constellations used at the end nodes}	
\label{fig:nodes}	
\end{figure}

The singular fade states are
\begin{align}
\nonumber
\gamma=1 &;\hspace{10pt} \theta=0,+\pi/2,-\pi/2,\pi\\
\nonumber
\gamma=1/\sqrt{2} &;\hspace{10pt} \theta=+\pi/4,+3\pi/4,-\pi/4,-3\pi/4\\
\nonumber
\gamma=\sqrt{2}&;\hspace{10pt}\theta=+\pi/4,+3\pi/4,-\pi/4,-3\pi/4.
\end{align}

We remove singular fade states one by one. Consider first, the case $(\gamma=1, \theta=0).$  
The singularity-removal constraints  are 
\begin{center}
$\{(0,1)(1,0)\};~~ \{(0,2)(1,3)(2,0)(3,1)\};~~ \{(0,3)(3,0)\};$ \\
$ \{(1,2)(2,1)\}; ~~\{(2,3)(3,2)\}.$
\end{center}

Satisfying these constraints,  a Latin Square can be constructed with $t$=4, in three different ways, $L_{1}, L_{2}$ and $L_{3}$ as shown in figure \ref{L1}, figure \ref{L2} and figure \ref{L3}. All these three clusterings corresponding to each Latin Square give the same performance on the basis of first minimum cluster distance in the MA phase. But the advantage with the one shown in figure \ref{L1} is that it removes singular fade state at $(\gamma=1, \theta=\pi)$ also. This is easily verified, by seeing that  after two cyclic shifts in the columns of $L_1$ the clustering that it results in is the same as the old one. This is explicitly  shown in Fig \ref{fig:latinshift}. 


The singularity-removal constraints for $(\gamma=1, \theta=\pi)$ are 
\begin{center}
$\{(0,3)(1,2)\}; ~~ \{(0,0)(1,1)(2,2)(3,3)\}; ~~ \{(0,1)(3,2)\};$ \\
$\{(1,0)(2,3)\}; ~~\{(2,1)(3,0)\}.$
\end{center}
\noindent
 The Latin squares to remove this singular fade state, $L_{1}, L_{4}$ and $L_{5}$ are shown in figure \ref{L1}, figure \ref{L4} and figure \ref{L5} respectively.
 
%
 
\begin{figure}
\centering
\vspace{0 cm}
\includegraphics[totalheight=1.4in,width=3in]{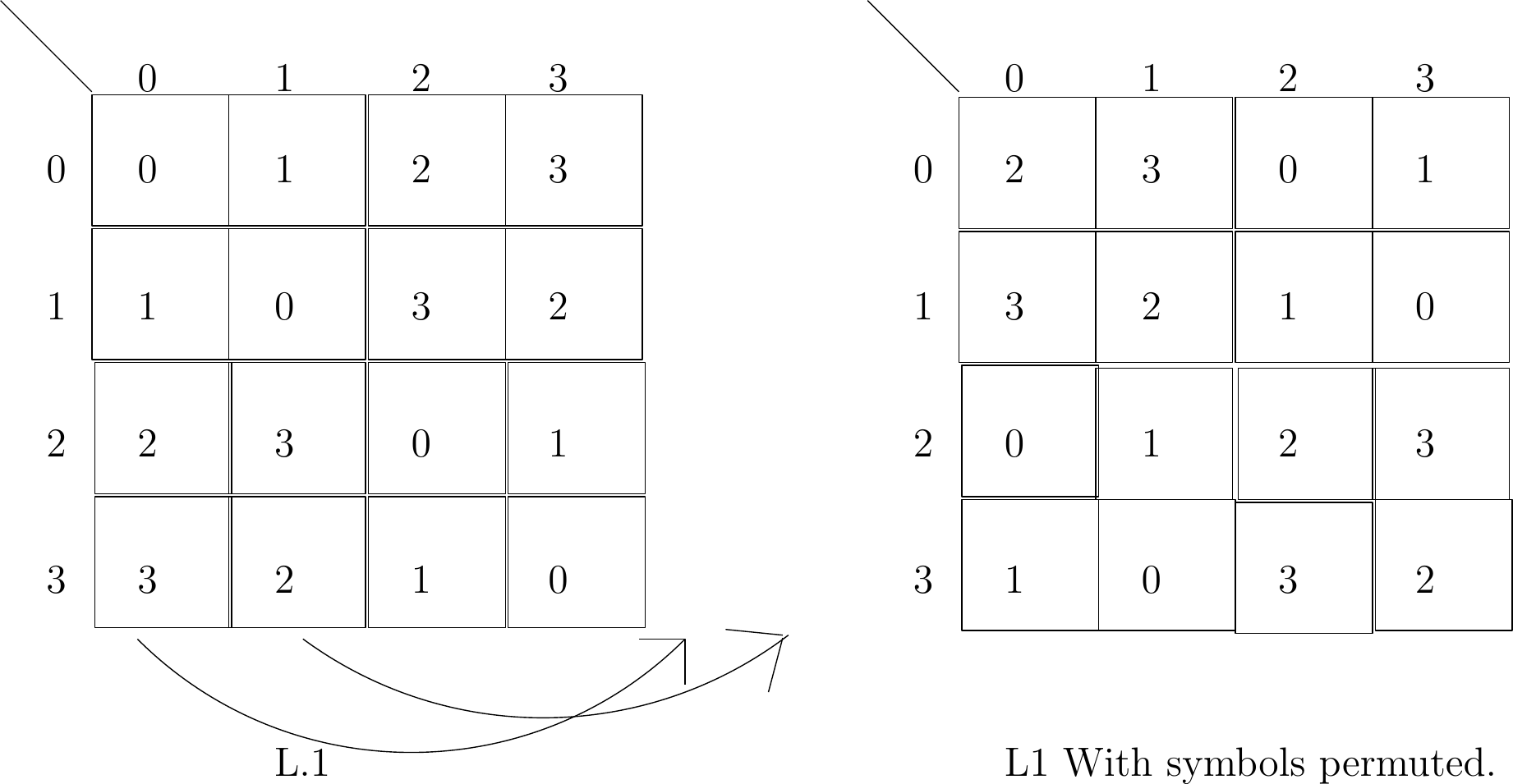}
\vspace{0 cm}
\caption{Two rotations of $L_{1}$ gives same clustering as that of $L_{1}$}	
\label{fig:latinshift}	
\end{figure}

In order to reduce the total number of different clusterings we select the clustering corresponds to the Latin Square $L_{1}$ shown in figure \ref{L1} as the clustering to remove both these singular fade states. The corresponding clustering, $\mathcal{C}_0$ is
\begin{center}
 \{(0,1)(1,0)(2,3)(3,2)\}, ~~ \{(0,2)(1,3)(2,0)(3,1)\}, \\
\{(0,3)(3,0)(1,2)(2,1)\}, ~~ \{(0,0)(1,1)(2,2)(3,3)\}.
\end{center}

Now consider the singular fade state $(\gamma=1, \theta=\pi/2).$ The singularity-removal constraints are 
\begin{center}
$\{(0,0)(1,3)\};~~\{(0,1)(1,2)(2,3)(3,0)\}; ~~\{(0,2)(3,3)\}$\\
$\{(1,1)(2,0)\}; ~~\{(2,2)(3,1)\}$
\end{center}
In this case also Latin Square can be constructed with $t=4$, in three different ways as shown in figure \ref{L6}, figure \ref{L7} and figure \ref{L8}. But as in earlier case out of these three one, $L_{6}$ (shown in figure \ref{L6}) will remove singular fade state $(\gamma=1, \theta=-\pi/2)$. The singularity-removal constraints for $(\gamma=1, \theta=-\pi/2)$ are 
\begin{center}
$\{(0,2)(1,1)\}; ~~\{(0,3)(1,0)(2,1)(3,2)\}; ~~\{(0,0)(3,1)\};$ \\
$\{(1,3)(2,2)\}; ~~ \{(2,0)(3,3)\}.$
\end{center}


 All the Latin Squares which remove the singular fade state $(\gamma=1, \theta=-\pi/2)$ are shown in figure \ref{L6}, figure \ref{L9} and figure \ref{L10}. We will select that clustering which reduces total number of different clusterings, i.e, $L_{6}$ as was done before. The corresponding clustering, $\mathcal{C}_1$ is as follows:

\begin{center}
\{(0,0)(1,3)(2,2)(3,1)\}, ~~ \{(0,1)(1,2)(2,3)(3,0)\}, \\

\{(0,2)(3,3)(1,1)(2,0)\}, ~~ \{(0,3)(1,0)(2,1)(3,2)\}.
\end{center}

The interesting point here is that the Latin Squares $L_{1}$ and $L_{6}$  are Isotopic Latin Squares. That is, clustering corresponding to $L_{6}$ is obtained by, cyclically shifting the columns of $L_{1}$ (since columns are indexed by constellation points used by node B) as in Fig.\ref{fig:Obtaining L.6 from L.1 by column shifting}.

\begin{figure}[t]
\centering
\vspace{-1 cm}
\includegraphics[totalheight=2in,width=3in]{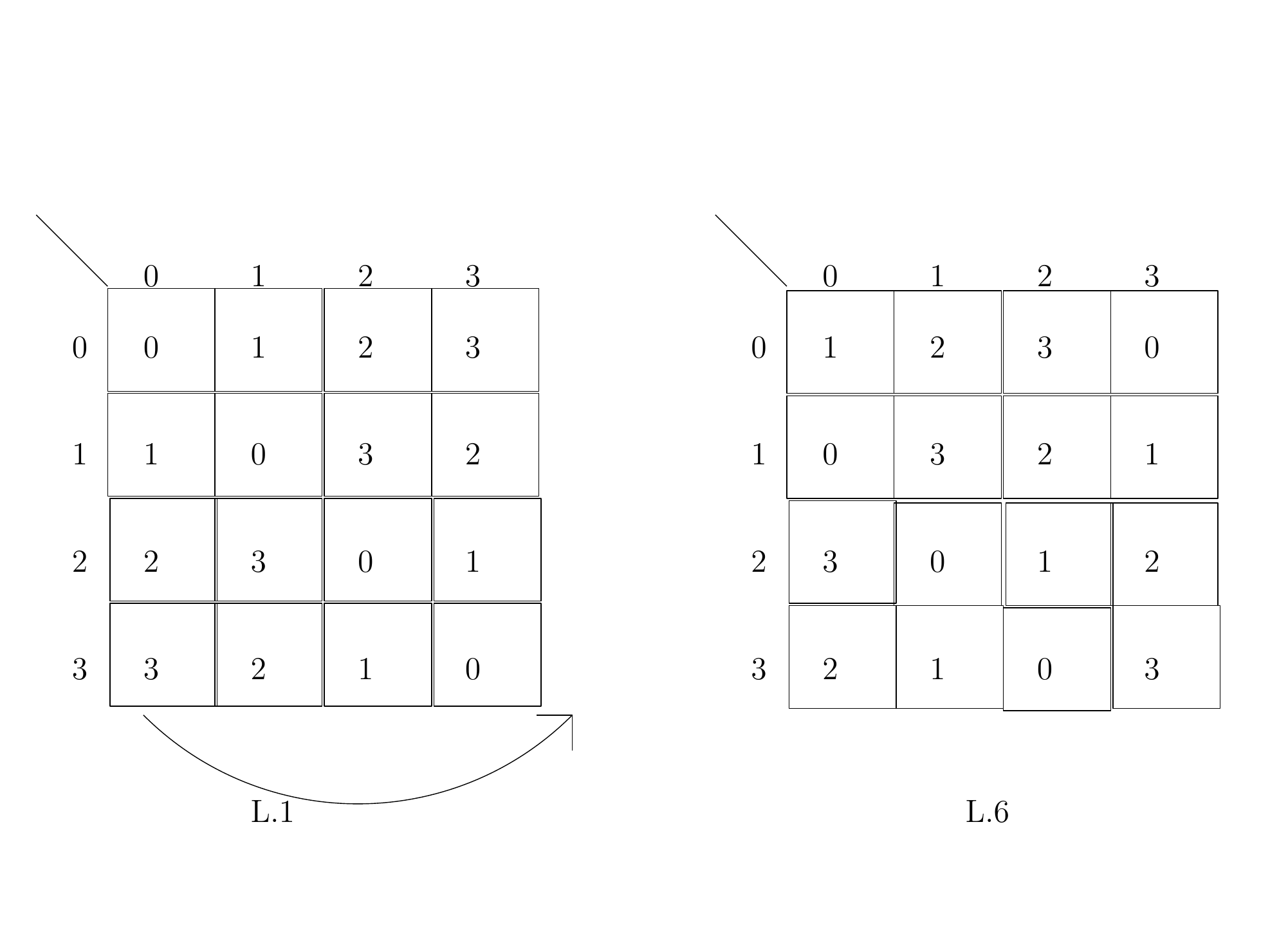}
\vspace{-.8 cm}
\caption{Obtaining $L_{6}$ from $L_{1}$ by column shifting}	
\label{fig:Obtaining L.6 from L.1 by column shifting}	
\end{figure}

Now we have removed four singular fade states till now all of them on the unit circle. Consider next,  $(\gamma=1/\sqrt{2}, \theta=\pi/4).$ The singularity-removal constraints are  
\begin{center}
$\{(0,1)(1,3)\};~~\{(0,2)(3,0)\};~~\{(1,2)(2,0)\};~~\{(2,3)(3,1)\}.$
\end{center}
The corresponding partially filled Latin Square is shown in figure {\ref{parfil1}.

It cannot be completed with $t$=4. That means that the relay has to use a constellation of size more than four. We can see that it can be completed with $t$=5. We get two clusterings as given in figure \ref{L11} and figure \ref{L12}. Both will remove this singular fade state and use constellation of size five. 

The clustering corresponding to Latin Square $L_{11}$, denoted by $\mathcal{C}_2$ is 
\begin{center}

\{(0,0)(2,3)(3,1)\},~~ \{(0,1)(1,3)(2,2)\},~~ \{(0,2)(1,1)(3,0)\},

\{(0,3)(1,0)(2,1)(3,2)\},~~ \{(1,2)(2,0)(3,3)\}.
\end{center}

The clustering corresponding to Latin Square $L_{12}$, denoted by $\mathcal{C}_3$  is 
\begin{center}
\{(0,0)(1,1)(2,2)(3,3)\},~~ \{(0,1)(1,3)(3,2)\}, ~~ \{(0,2)(2,1)(3,0)\},

\{(0,3)(1,2)(2,0)\},~~ \{(1,0)(2,3)(3,1)\}\}.
\end{center}

Now considering the next singular fade state $(\gamma=\sqrt{2}, \theta=\pi/4),$ by the same procedure as before, the singularity-removal constraints are 
\begin{center}
\{(0,1)(2,0)\},~ \{(0,2)(2,3)\},~\{(1,2)(3,1)\},~\{(1,3)(3,0)\}.
\end{center}
The partially filled Latin Square is shown in figure \ref{parfil2}. This cannot be completed with $t=4$, but by $t=5$ it can be completed in two ways as in $L_{13}$ and $L_{14}$ shown in figure \ref{L13} and figure \ref{L14}. The corresponding clusterings are shown in the Table \ref{table1}.
\begin{figure}[h]
\centering
\subfigure[]{
{
\begin{tabular}{|c|c|c|c|}
\hline 0 & 1 & 2 & 3\\ 
\hline 1 & 0 & 3 & 2\\ 
\hline 2 & 3 & 0 & 1\\ 
\hline 3 & 2 & 1 & 0\\ 
\hline 
\end{tabular}}
\label{L1}
}
\subfigure[]{
{\begin{tabular}{|c|c|c|c|}
\hline 0 & 1 & 2 & 3\\ 
\hline 1 & 3 & 0 & 2\\ 
\hline 2 & 0 & 3 & 1\\ 
\hline 3 & 2 & 1 & 0\\ 
\hline
\end{tabular}} 
\label{L2}
}
\subfigure[]{
{
\begin{tabular}{|c|c|c|c|}
\hline 0 & 1 & 2 & 3\\ 
\hline 1 & 0 & 3 & 2\\ 
\hline 2 & 3 & 1 & 0\\ 
\hline 3 & 2 & 0 & 1\\ 
\hline
\end{tabular}}
\label{L3}
}
\caption[]{Completely Filled Latin Squares for $\gamma=1, \theta=0$ }
\end{figure}
\begin{figure}[h]
\centering
\subfigure[]{
{
 \begin{tabular}{|c|c|c|c|}
\hline 0 & 1 & 2 & 3\\ 
\hline 2 & 0 & 3 & 1\\ 
\hline 1 & 3 & 0 & 2\\ 
\hline 3 & 2 & 1 & 0\\ 
\hline 
\end{tabular}}
\label{L4}
}
\subfigure[]{
{\begin{tabular}{|c|c|c|c|}
\hline 0 & 1 & 2 & 3\\ 
\hline 1 & 0 & 3 & 2\\ 
\hline 3 & 2 & 0 & 1\\ 
\hline 2 & 3 & 1 & 0\\ 
\hline
\end{tabular}} 
\label{L5}
}

\caption[]{Completely Filled Latin Squares for $\gamma=1, \theta=\pi$ }
\end{figure}
\begin{figure}[h]
\centering
\subfigure[]{
{
\begin{tabular}{|c|c|c|c|}
\hline 1 & 2 & 3 & 0\\ 
\hline 0 & 3 & 2 & 1\\ 
\hline 3 & 0 & 1 & 2\\ 
\hline 2 & 1 & 0 & 3\\ 
\hline 
\end{tabular}}
\label{L6}
}
\subfigure[]{
{\begin{tabular}{|c|c|c|c|}
\hline 1 & 2 & 3 & 0\\ 
\hline 3 & 0 & 2 & 1\\ 
\hline 0 & 3 & 1 & 2\\ 
\hline 2 & 1 & 0 & 3\\ 
\hline
\end{tabular}} 
\label{L7}
}
\subfigure[]{
{
\begin{tabular}{|c|c|c|c|}
\hline 1 & 2 & 3 & 0\\ 
\hline 0 & 3 & 2 & 1\\ 
\hline 3 & 1 & 0 & 2\\ 
\hline 2 & 0 & 1 & 3\\ 
\hline
\end{tabular}}
\label{L8}
}
\caption[]{Completely Filled Latin Squares for $\gamma=1, \theta=\pi/2$ }
\end{figure}
\begin{figure}[h]
\centering
\subfigure[]{
{
 \begin{tabular}{|c|c|c|c|}
\hline 1 & 2 & 3 & 0\\ 
\hline 0 & 3 & 1 & 2\\ 
\hline 3 & 0 & 2 & 1\\ 
\hline 2 & 1 & 0 & 3\\ 
\hline 
\end{tabular}}
\label{L9}
}
\subfigure[]{
{\begin{tabular}{|c|c|c|c|}
\hline 1 & 2 & 3 & 0\\ 
\hline 0 & 3 & 2 & 1\\ 
\hline 2 & 0 & 1 & 3\\ 
\hline 3 & 1 & 0 & 2\\ 
\hline
\end{tabular}} 
\label{L10}
}

\caption[]{Completely Filled Latin Squares for $\gamma=1, \theta=-\pi/2$ }
\end{figure}
\begin{figure}[h]
\centering
\subfigure[]{
{
 \begin{tabular}{|c|c|c|c|}
\hline   & 0 & 1 &  \\ 
\hline   &   & 2 & 0\\ 
\hline 2 &   &   & 3\\ 
\hline 1 & 3 &   &  \\ 
\hline 
\end{tabular}}
\label{parfil1}
}
\subfigure[]{
{\begin{tabular}{|c|c|c|c|}
\hline   & 0 & 1 &  \\ 
\hline   &   & 2 & 3\\ 
\hline 0 &   &   & 1\\ 
\hline 3 & 2 &   &  \\ 
\hline 
\end{tabular} } 
\label{parfil2}
}

\caption[]{Partially Filled Latin Squares for $\gamma=1/\sqrt{2}, \theta=\pi/4$ and  for $\gamma=\sqrt{2}, \theta=\pi/4$ }
\end{figure}
\begin{figure}[h]
\centering
\subfigure[]{
{
\begin{tabular}{|c|c|c|c|}
\hline 3 & 0 & 1 & 4\\ 
\hline 4 & 1 & 2 & 0\\ 
\hline 2 & 4 & 0 & 3\\ 
\hline 1 & 3 & 4 & 2\\ 
\hline 
\end{tabular}}
\label{L11}
}
\subfigure[]{
{\begin{tabular}{|c|c|c|c|}
\hline 4 & 0 & 1 & 2\\ 
\hline 3 & 4 & 2 & 0\\ 
\hline 2 & 1 & 4 & 3\\ 
\hline 1 & 3 & 0 & 4\\ 
\hline
\end{tabular} } 
\label{L12}
}

\caption[]{$L_{11}$ for $\gamma=1/\sqrt{2}, \theta=\pi/4$ and $\gamma=\sqrt{2}, \theta=3\pi/4$ and $L_{12}$ for $\gamma=1/\sqrt{2}, \theta=\pi/4$ and $\gamma=\sqrt{2}, \theta=-\pi/4$}
\end{figure}
\begin{figure}[h]
\centering
\subfigure[]{
{
\begin{tabular}{|c|c|c|c|}
\hline 4 & 0 & 1 & 2\\ 
\hline 1 & 4 & 2 & 3\\ 
\hline 0 & 3 & 4 & 1\\ 
\hline 3 & 2 & 0 & 4\\ 
\hline 
\end{tabular}}
\label{L13}
}
\subfigure[]{
{
\begin{tabular}{|c|c|c|c|}
\hline 2 & 0 & 1 & 4\\ 
\hline 4 & 1 & 2 & 3\\ 
\hline 0 & 4 & 3 & 1\\ 
\hline 3 & 2 & 4 & 0\\ 
\hline
\end{tabular}}
\label{L14}
}
\caption[]{$L_{13}$ for $\gamma=\sqrt{2}, \theta=\pi/4$ and $\gamma=1/\sqrt{2}, \theta=-\pi/4$  and $L_{14}$ for $\gamma=\sqrt{2}, \theta=\pi/4$ and $\gamma=1/\sqrt{2}, \theta=3\pi/4$}
\end{figure}
\begin{figure}[h]
\centering
\subfigure[]{
{
\begin{tabular}{|c|c|c|c|}
\hline 0 & 1 & 2 & 3\\ 
\hline 1 & 4 & 0 & 2\\ 
\hline 2 & 0 & 3 & 4\\ 
\hline 3 & 2 & 4 & 1\\ 
\hline 
\end{tabular} }
\label{L15}
}
\subfigure[]{
{
\begin{tabular}{|c|c|c|c|}
\hline 0 & 1 & 2 & 3\\ 
\hline 1 & 3 & 4 & 2\\ 
\hline 2 & 4 & 1 & 0\\ 
\hline 3 & 2 & 0 & 4\\ 
\hline
\end{tabular}}
\label{L16}
}
\caption[]{$L_{15}$ for $\gamma=\sqrt{2}, \theta=-3\pi/4$ and $\gamma=1/\sqrt{2}, \theta=3\pi/4$  and $L_{16}$ for $\gamma=\sqrt{2}, \theta=3\pi/4$ and $\gamma=1/\sqrt{2}, \theta=-3\pi/4$}
\end{figure}
\begin{figure}[h]
\centering
\subfigure[]{
{
 \begin{tabular}{|c|c|c|c|}
\hline 0 & 1 & 2 & 3\\ 
\hline 4 & 3 & 1 & 0\\ 
\hline 3 & 2 & 4 & 1\\ 
\hline 1 & 4 & 0 & 2\\ 
\hline 
\end{tabular}  }
\label{L17}
}
\subfigure[]{
{
\begin{tabular}{|c|c|c|c|}
\hline 0 & 1 & 2 & 3\\ 
\hline 2 & 4 & 1 & 0\\ 
\hline 4 & 0 & 3 & 1\\ 
\hline 1 & 3 & 4 & 2\\ 
\hline
\end{tabular}}
\label{L18}
}
\caption[]{$L_{17}$ for $\gamma=\sqrt{2}, \theta=-\pi/4$ and $\gamma=1/\sqrt{2}, \theta=-3\pi/4$ and $L_{18}$ for $\gamma=\sqrt{2}, \theta=-3\pi/4$ and $\gamma=1/\sqrt{2}, \theta=-\pi/4$}
\end{figure}
\begin{table*}
\centering
 \caption{Clusterings Obtained for different singular fade states when the end nodes use QPSK constellations}
 \label{table1}
 \begin{tabular}{|c|c|c|c|}

\hline Sl.No & Singular fade states & Clustering & Cluster\\ 
\hline 1& $\gamma=1,\theta=0$                                           & $\mathcal{C}_0$ &$\{$\{(0,1)(1,0)(2,3)(3,2)\}$,$\{(0,2)(1,3)(2,0)(3,1)\}$,$\{(0,3)(3,0)(1,2)(2,1)\}$,$\{(0,0)(1,1)(2,2)(3,3)\}$\}$\\ 
\hline 2&$\gamma=1,\theta=\pi/2$                                           & $\mathcal{C}_1$ &$\{$\{(0,0)(1,3)(2,2)(3,1)\}$,$\{(0,1)(1,2)(2,3)(3,0)\}$,$\{(0,2)(3,3)(1,1)(2,0)$\},$\{(0,3)(1,0)(2,1)(3,2)$\}\}$\\ 
\hline 3&$\gamma=1,\theta=\pi$                                           & $\mathcal{C}_0$ &$\{$\{(0,1)(1,0)(2,3)(3,2)\}$,$\{(0,2)(1,3)(2,0)(3,1)\}$,$\{(0,3)(3,0)(1,2)(2,1)\}$,$\{(0,0)(1,1)(2,2)(3,3)\}$\}$\\ 
\hline 4&$\gamma=1,\theta=-\pi/2$                                           & $\mathcal{C}_1$ &$\{$\{(0,0)(1,3)(2,2)(3,1)\}$,$\{(0,1)(1,2)(2,3)(3,0)\}$,$\{(0,2)(3,3)(1,1)(2,0)$\},$\{(0,3)(1,0)(2,1)(3,2)$\}\}$\\ 
\hline 5.a&$\gamma=1/\sqrt{2},\theta=\pi/4$                                           & $\mathcal{C}_2$ &$\{$\{(0,0)(2,3)(3,1)\}$,$\{(0,1)(1,3)(2,2)\}$,$\{(0,2)(1,1)(3,0)\}$,$\{(0,3)(1,0)(2,1)(3,2)$\},$\{(1,2)(2,0)(3,3)$\}\}$\\ 
 5.b&                                           & $\mathcal{C}_3$ &$\{$\{(0,0)(1,1)(2,2)(3,3)\}$,$\{(0,1)(1,3)(3,2)\}$,$\{(0,2)(2,1)(3,0)\}$,$\{(0,3)(1,2)(2,0)$\},$\{(1,0)(2,3)(3,1)$\}\}$\\ 
\hline 6.a&$\gamma=\sqrt{2},\theta=\pi/4$                                           & $\mathcal{C}_4$ &$\{$\{(0,0)(1,1)(2,2)(3,3)\}$,$\{(0,1)(2,0)(3,2)\}$,$\{(0,2)(1,0)(2,3)\}$,$\{(0,3)(1,2)(3,1)$\},$\{(1,3)(2,1)(3,0)$\}\}$\\ 
 6.b&                                          & $\mathcal{C}_5$ &$\{$\{(0,0)(1,2)(3,1)\}$,$\{(0,1)(2,0)(3,3)\}$,$\{(0,2)(1,1)(2,3)\}$,$\{(0,3)(1,0)(2,1)(3,2)$\},$\{(1,3)(2,2)(3,0)$\}\}$\\ 
\hline 7.a&$\gamma=1/\sqrt{2},\theta=3\pi/4$                                           & $\mathcal{C}_6$ &$\{$\{(0,0)(1,2)(2,1)\}$,$\{(0,1)(1,0)(3,3)\}$,$\{(0,2)(1,3)(2,0)(3,1)\}$,$\{(0,3)(2,2)(3,0)$\},$\{(1,1)(2,3)(3,2)$\}\}$\\ 
 7.b&                                          & $\mathcal{C}_5$ &$\{$\{(0,0)(1,2)(3,1)\}$,$\{(0,1)(2,0)(3,3)\}$,$\{(0,2)(1,1)(2,3)\}$,$\{(0,3)(1,0)(2,1)(3,2)$\},$\{(1,3)(2,2)(3,0)$\}\}$\\ 
\hline 8.a&$\gamma=\sqrt{2},\theta=3\pi/4$                                           & $\mathcal{C}_2$ &$\{$\{(0,0)(2,3)(3,1)\}$,$\{(0,1)(1,3)(2,2)\}$,$\{(0,2)(1,1)(3,0)\}$,$\{(0,3)(1,0)(2,1)(3,2)$\},$\{(1,2)(2,0)(3,3)$\}\}$\\ 
 8.b&                                           & $\mathcal{C}_7$ &$\{$\{(0,0)(2,3)(3,2)\}$,$\{(0,1)(1,0)(2,2)\}$,$\{(0,2)(1,3)(2,0)(3,1)\}$,$\{(0,3)(1,1)(3,0)$\},$\{(1,2)(2,1)(3,3)$\}\}$\\ 
\hline 9.a&$\gamma=1/\sqrt{2},\theta=-3\pi/4$                                           & $\mathcal{C}_7$ &$\{$\{(0,0)(2,3)(3,2)\}$,$\{(0,1)(1,0)(2,2)\}$,$\{(0,2)(1,3)(2,0)(3,1)\}$,$\{(0,3)(1,1)(3,0)$\},$\{(1,2)(2,1)(3,3)$\}\}$\\
 9.b &                                          & $\mathcal{C}_8$ &$\{$\{(0,0)(1,3)(3,2)\}$,$\{(0,1)(1,2)(2,3)(3,0)\}$,$\{(0,2)(2,1)(3,3)\}$,$\{(0,3)(1,1)(2,0)$\},$\{(1,0)(2,2)(3,1)$\}\}$\\
\hline 10.a&$\gamma=\sqrt{2},\theta=-3\pi/4$                                           & $\mathcal{C}_6$ &$\{$\{(0,0)(1,2)(2,1)\}$,$\{(0,1)(1,0)(3,3)\}$,$\{(0,2)(1,3)(2,0)(3,1)\}$,$\{(0,3)(2,2)(3,0)$\},$\{(1,1)(2,3)(3,2)$\}\}$\\ 
 10.b&                                         & $\mathcal{C}_9$ &$\{$\{(0,0)(1,3)(2,1)\}$,$\{(0,1)(1,2)(2,3)(3,0)\}$,$\{(0,2)(1,0)(3,3)\}$,$\{(0,3)(2,2)(3,1)$\},$\{(1,1)(2,0)(3,2)$\}\}$\\ 
\hline 11.a&$\gamma=1/\sqrt{2},\theta=-\pi/4$                                           & $\mathcal{C}_4$ &$\{$\{(0,0)(1,1)(2,2)(3,3)\}$,$\{(0,1)(2,0)(3,2)\}$,$\{(0,2)(1,0)(2,3)\}$,$\{(0,3)(1,2)(3,1)$\},$\{(1,3)(2,1)(3,0)$\}\}$\\ 
 11.b&                                         & $\mathcal{C}_9$ &$\{$\{(0,0)(1,3)(2,1)\}$,$\{(0,1)(1,2)(2,3)(3,0)\}$,$\{(0,2)(1,0)(3,3)\}$,$\{(0,3)(2,2)(3,1)$\},$\{(1,1)(2,0)(3,2)$\}\}$\\
\hline 12.a&$\gamma=\sqrt{2},\theta=-\pi/4$                                           & $\mathcal{C}_3$ &$\{$\{(0,0)(1,1)(2,2)(3,3)\}$,$\{(0,1)(1,3)(3,2)\}$,$\{(0,2)(2,1)(3,0)\}$,$\{(0,3)(1,2)(2,0)$\},$\{(1,0)(2,3)(3,1)$\}\}$\\  
 12.b&                                         & $\mathcal{C}_8$ &$\{$\{(0,0)(1,3)(3,2)\}$,$\{(0,1)(1,2)(2,3)(3,0)\}$,$\{(0,2)(2,1)(3,3)\}$,$\{(0,3)(1,1)(2,0)$\},$\{(1,0)(2,2)(3,1)$\}\}$\\
\hline 
\end{tabular} 
\vspace{-.1 cm}
\end{table*}


Next, consider the singular fade state $(\gamma=1/\sqrt{2}, \theta=3\pi/4).$ The singularity-removal constraints are 
\begin{center}
\{(0,0)(1,2)\},~\{(0,1)(3,3)\},~\{(1,1)(2,3)\},~\{(2,2)(3,0)\}.
\end{center}
The corresponding Latin Squares are shown in figure \ref{L14} and figure \ref{L15}. 

For the singular fade state $(\gamma=\sqrt{2}, \theta=3\pi/4)$, the singularity-removal constraints are 
\begin{center}
\{(0,0)(2,3)\},~\{(0,1)(2,2)\},~\{(1,1)(3,0)\},~\{(1,2)(3,3)\}
\end{center}
and the corresponding Latin Squares are shown in figure \ref{L11} and figure \ref{L16} .

Similarly, for the singular fade state $(\gamma=1/\sqrt{2}, \theta=-3\pi/4),$ the singularity-removal constraints are 
\begin{center}
\{(0,0)(3,2)\},~\{(0,3)(1,1)\},~\{(1,0)(2,2)\},~\{(2,1)(3,3)\}
\end{center}
with the  corresponding  Latin Squares as shown in  figure \ref{L16} and figure \ref{L17}.

The singularity-removal constraints for singular fade state $(\gamma=\sqrt{2}, \theta=-3\pi/4)$ are 
\begin{center}
\{(0,0)(2,1)\},~\{(0,3)(2,2)\},~\{(1,0)(3,3)\},~\{(1,1)(3,2)\}.
\end{center}
 and the Latin Squares are given in figure \ref{L15} and  figure \ref{L18}. The singularity-removal constraints for singular fade state $(\gamma=1/\sqrt{2}, \theta=-\pi/4)$ are 
\begin{center} 
\{(0,2)(1,0)\},~\{(0,3)(3,1)\},~\{(1,3)(2,1)\},~\{(2,0)(3,2)\}.
\end{center}
The Latin Squares are given in figure.\ref{L13} and  figure.\ref{L18}. The singularity-removal constraints for singular fade state $\gamma=\sqrt{2}$ and $\theta=-\pi/4$ are 
\begin{center}
\{(0,2)(2,1)\},~\{(0,3)(2,0)\},~\{(1,0)(3,1)\},~\{(1,3)(3,2)\}.
\end{center}
The Latin Squares are given in figure \ref{L12} and  figure \ref{L17} .

It is observed that to remove all other singular fade states not lying on unit circle the relay needs a constellation of size five. Table \ref{table1} shows the singular fade states and the corresponding clusterings. There are two clusterings to remove a singular fade state for all singular fade states except for those with  $\gamma=1.$  We can select any one. Anyone from the two $\{\mathcal{C}_2,\mathcal{C}_3 \}$ can be selected to remove singular fade state $(\gamma=1/\sqrt{2}, \theta=\pi/4).$ After that, by column permutations we can remove the singular fade states with $(\gamma=1/\sqrt{2},$ and $\theta=+3\pi/4,-\pi/4,-3\pi/4$. By taking transpose of the Latin Square for $(\gamma=1/\sqrt{2}, \theta=\pi/4)$ we can remove singular fade state $(\gamma=\sqrt{2}, \theta=-\pi/4).$  After that by column permutations we can remove the singular fade states with $\gamma=\sqrt{2}$ and $\theta=+3\pi/4,+\pi/4,-3\pi/4$. If we select $\mathcal{C}_2$ to remove $(\gamma=1/\sqrt2, \theta=\pi/4),$ we will get the following set of clusterings $\{\mathcal{C}_0,\mathcal{C}_1,\mathcal{C}_2,\mathcal{C}_4,\mathcal{C}_6,\mathcal{C}_8\}$ to remove all the singular fade states. In the other case, when we select $\mathcal{C}_3$ to remove $(\gamma=1/\sqrt{2},\theta=\pi/4)$ we will get the following set of clusterings $\{\mathcal{C}_0,\mathcal{C}_1,\mathcal{C}_3,\mathcal{C}_5,\mathcal{C}_7,\mathcal{C}_9\}$ to remove all the singular fade states.

\subsection{End nodes use 8PSK}
Consider the scenario where the end nodes use 8-PSK constellation as shown in Fig.\ref{8psk}.

\begin{figure}[htbp]
\centering
\includegraphics[totalheight=1.5in,width=1.5in]{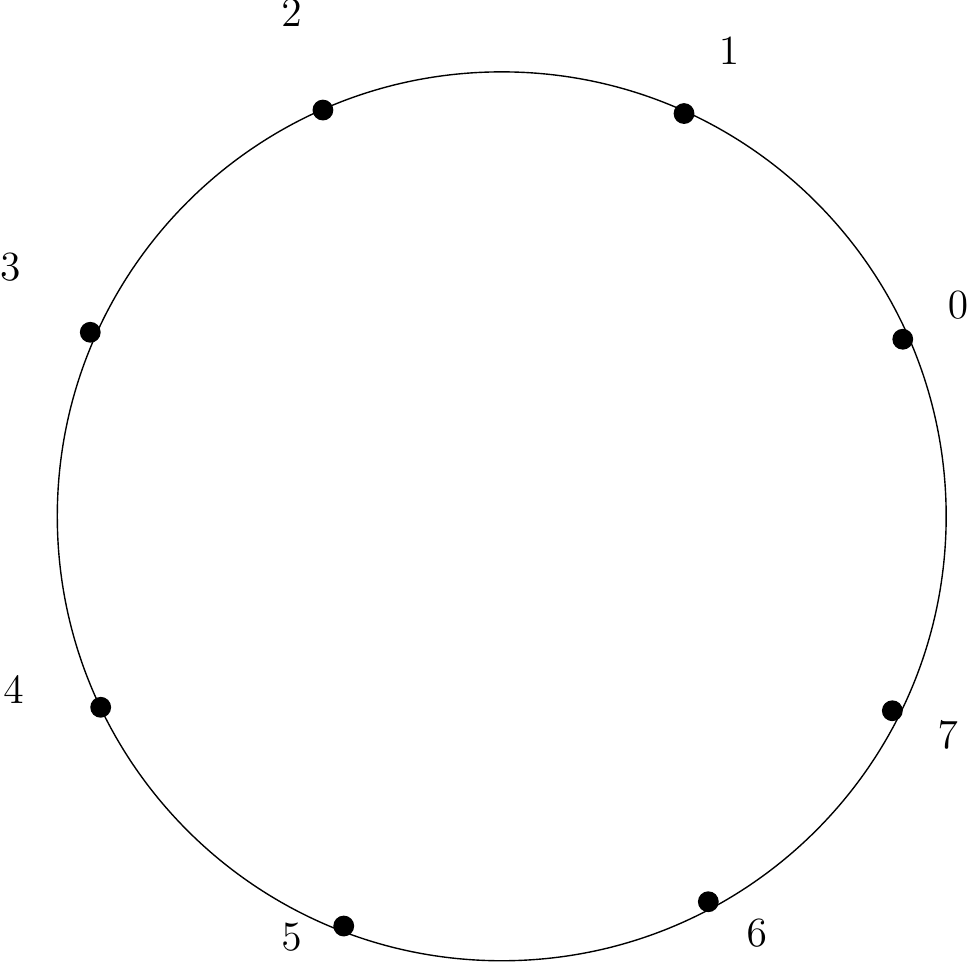}
\caption{8PSK constellation used by nodes A and B}	
\label{8psk}	
\end{figure}

For 8-PSK signal set, the 104 singular fade states are shown in Fig.\ref{8psk_sing}. These singular fade states lie in 13 circles, 8 in each circle separated by an angle $\pi/4$. It is sufficient to get Latin Squares for these 13 values, then by isotopic property of the Latin Squares all other singular fade states are removed. Further by taking transposes of the Latin Square for a singular fade state $\gamma, \theta$, another singular fade state $1/\gamma, -\theta$ is removed. As discussed XOR removes singular fade state $\gamma=1, \theta=0$. Finally it is seen that obtaining six Latin Squares for different singular fade state values for $\gamma < 1$ (or $\gamma > 1$) and the XOR is enough to remove all the 104 singular fade states. The 104 singular fade states are,\\
for $\theta=2m\pi/8$ , where $m=0,1,2\cdots7$\\\\
$\gamma=\dfrac{\sin\pi/8}{\sin\pi/8},\dfrac{\sin\pi/8}{\sin3\pi/8}, \dfrac{\sin3\pi/8}{\sin\pi/8}, \dfrac{\sin4\pi/8}{\sin2\pi/8}$ and $\dfrac{\sin2\pi/8}{\sin4\pi/8}$.\\\\
for $\theta=(2m+1)\pi/8$ , where $m=0,1,2\cdots7$\\\\
$\gamma=\dfrac{\sin\pi/8}{\sin4\pi/8},\dfrac{\sin\pi/8}{\sin2\pi/8}, \dfrac{\sin2\pi/8}{\sin3\pi/8}, \dfrac{\sin3\pi/8}{\sin4\pi/8},\dfrac{\sin4\pi/8}{\sin3\pi/8},$\\
\hspace*{.75cm}$ \dfrac{\sin3\pi/8}{\sin2\pi/8}, \dfrac{\sin2\pi/8}{\sin\pi/8}$ and $\dfrac{\sin4\pi/8}{\sin\pi/8}$.

%
%

We consider singular fade states one by one. The singularity constraints for $\gamma=1, \theta=0$ are 
\begin{center}
$\{(0,1)(1,0)\}$, $\{(0,2)(2,0)\}$, $\{(0,3)(3,0)\}$, $\{(0,5)(5,0)\}$\\*
$\{(0,6)(6,0)\}$, $\{(0,7)(7,0)\}$, $\{(1,2)(2,1)\}$, $\{(1,3)(3,1)\}$\\*
$\{(1,4)(4,1)\}$, $\{(1,6)(6,1)\}$, $\{(1,7)(7,1)\}$, $\{(2,3)(3,2)\}$\\*
$\{(2,4)(4,2)\}$, $\{(2,6)(6,2)\}$, $\{(2,7)(7,2)\}$, $\{(3,4)(4,3)\}$\\*
$\{(3,5)(5,3)\}$, $\{(3,6)(6,3)\}$, $\{(4,5)(5,4)\}$, $\{(4,6)(6,4)\}$\\*
$\{(4,7)(7,4)\}$, $\{(5,6)(6,5)\}$, $\{(5,7)(7,5)\}$, $\{(6,7)(7,6)\}$\\*
$\{(0,4)(1,5)(2,6)(3,7)(4,0)(5,1)(6,2)(7,3)\}$\\*
\end{center}

XOR can remove this singular fade state. Further by appropriately column shifting all the singular fade states with $\gamma=1$ is removed. For example consider $\gamma=1, \theta=\pi/4$. The singularity constraints are 
\begin{center}
$\{(0,0)(1,7)\}$, $\{(0,1)(2,7)\}$, $\{(0,2)(3,7)\}$, $\{(0,4)(5,7)\}$\\*
$\{(0,5)(6,7)\}$, $\{(0,6)(7,7)\}$, $\{(1,1)(2,0)\}$, $\{(1,2)(3,0)\}$\\*
$\{(1,3)(4,0)\}$, $\{(1,5)(6,0)\}$, $\{(1,6)(7,0)\}$, $\{(2,2)(3,1)\}$\\*
$\{(2,3)(4,1)\}$, $\{(2,5)(6,1)\}$, $\{(2,6)(7,1)\}$, $\{(3,3)(4,2)\}$\\*
$\{(3,4)(5,2)\}$, $\{(3,5)(6,2)\}$, $\{(4,4)(5,3)\}$, $\{(4,5)(6,3)\}$\\*
$\{(4,6)(7,3)\}$, $\{(5,5)(6,4)\}$, $\{(5,6)(7,4)\}$, $\{(6,6)(7,5)\}$\\*
$\{(0,3)(1,4)(2,5)(3,6)(4,7)(5,0)(6,1)(7,2)\}$\\*
\end{center}

Next consider $\gamma=\dfrac{\sin \pi/8}{\sin 3\pi/8} =0.414, \theta=0$. The singularity constraints are 

\begin{center}
$\{(0,2)(1,7)\}$, $\{(0,3)(1,6)\}$, $\{(0,5)(7,2)\}$, $\{(0,6)(7,1)\}$\\*
$\{(1,3)(2,0)\}$, $\{(1,4)(2,7)\}$, $\{(2,4)(3,1)\}$, $\{(2,5)(3,0)\}$\\*
$\{(3,5)(4,2)\}$, $\{(3,6)(4,1)\}$, $\{(4,6)(5,3)\}$, $\{(4,7)(5,2)\}$\\*
$\{(5,0)(6,3)\}$, $\{(5,7)(6,4)\}$, $\{(6,0)(7,5)\}$, $\{(6,1)(7,4)\}$\\*
\end{center}

\begin{figure}[h]
\centering
\subfigure[$\gamma=1, \theta=0$]{
{
\begin{tabular}{|c|c|c|c|c|c|c|c|}
\hline 0 & 1 & 2 & 3 & 4 & 5 & 6 & 7 \\ 
\hline 1 & 0 & 3 & 2 & 5 & 4 & 7 & 6 \\
\hline 2 & 3 & 0 & 1 & 6 & 7 & 4 & 5 \\  
\hline 3 & 2 & 1 & 0 & 7 & 6 & 5 & 4 \\ 
\hline 4 & 5 & 6 & 7 & 0 & 1 & 2 & 3 \\ 
\hline 5 & 4 & 7 & 6 & 1 & 0 & 3 & 2 \\
\hline 6 & 7 & 4 & 5 & 2 & 3 & 0 & 1 \\ 
\hline 7 & 6 & 5 & 4 & 3 & 2 & 1 & 0 \\
\hline 
\end{tabular}}
\label{8L1}
}
\subfigure[$\gamma=1, \theta=\pi/4$]{
{
\begin{tabular}{|c|c|c|c|c|c|c|c|}
\hline 1 & 2 & 3 & 4 & 5 & 6 & 7 & 0 \\ 
\hline 0 & 3 & 2 & 5 & 4 & 7 & 6 & 1 \\
\hline 3 & 0 & 1 & 6 & 7 & 4 & 5 & 2 \\  
\hline 2 & 1 & 0 & 7 & 6 & 5 & 4 & 3\\ 
\hline 5 & 6 & 7 & 0 & 1 & 2 & 3 & 4\\ 
\hline 4 & 7 & 6 & 1 & 0 & 3 & 2 & 5\\
\hline 7 & 4 & 5 & 2 & 3 & 0 & 1 & 6\\ 
\hline 6 & 5 & 4 & 3 & 2 & 1 & 0 & 7\\
\hline 
\end{tabular}}
\label{8L2}
}
\subfigure[$\gamma=\dfrac{\sin \pi/8}{\sin 3\pi/8}, \theta=0$]{
{
\begin{tabular}{|c|c|c|c|c|c|c|c|}
\hline 0 & 1 & 2 & 3 & 4 & 5 & 6 & 7 \\ 
\hline 5 & 4 & 7 & 6 & 1 & 0 & 3 & 2\\
\hline 6 & 3 & 0 & 5 & 2 & 7 & 4 & 1 \\  
\hline 7 & 2 & 1 & 4 & 3 & 6 & 5 & 0 \\ 
\hline 4 & 5 & 6 & 7 & 0 & 1 & 2 & 3 \\ 
\hline 1 & 0 & 3 & 2 & 5 & 4 & 7 & 6 \\
\hline 2 & 7 & 4 & 1 & 6 & 3 & 0 & 5 \\ 
\hline 3 & 6 & 5 & 0 & 7 & 2 & 1 & 4 \\
\hline 
\end{tabular}}
\label{8L3}
}
\subfigure[$\gamma=\dfrac{\sin \pi/8}{\sin 3\pi/8}, \theta=0$]{
{
\begin{tabular}{|c|c|c|c|c|c|c|c|}
\hline 1 & 2 & 3 & 4 & 5 & 6 & 7 & 0\\ 
\hline 4 & 7 & 6 & 1 & 0 & 3 & 2 & 5\\
\hline 3 & 0 & 5 & 2 & 7 & 4 & 1 & 6\\  
\hline 2 & 1 & 4 & 3 & 6 & 5 & 0 & 7\\ 
\hline 5 & 6 & 7 & 0 & 1 & 2 & 3 & 4\\ 
\hline 0 & 3 & 2 & 5 & 4 & 7 & 6 & 1\\
\hline 7 & 4 & 1 & 6 & 3 & 0 & 5 & 2\\ 
\hline 6 & 5 & 0 & 7 & 2 & 1 & 4 & 3\\
\hline 
\end{tabular}}
\label{8L4}
} 
\caption[]{Latin Squares Corresponding to Different singular fade states }
\end{figure}
\begin{figure}[h]
\centering
\subfigure[$\gamma=\dfrac{\sin 3\pi/8}{\sin \pi/8}, \theta=-\pi/4$]{
{
\begin{tabular}{|c|c|c|c|c|c|c|c|}
\hline 1 & 4 & 3 & 2 & 5 & 0 & 7 & 6 \\ 
\hline 2 & 7 & 0 & 1 & 6 & 3 & 4 & 5\\
\hline 3 & 6 & 5 & 4 & 7 & 2 & 1 & 0 \\  
\hline 4 & 1 & 2 & 3 & 0 & 5 & 6 & 7 \\ 
\hline 5 & 0 & 7 & 6 & 1 & 4 & 3 & 2 \\ 
\hline 6 & 3 & 4 & 5 & 2 & 7 & 0 & 1 \\
\hline 7 & 2 & 1 & 0 & 3 & 6 & 5 & 4 \\ 
\hline 0 & 5 & 6 & 7 & 4 & 1 & 2 & 3 \\
\hline 
\end{tabular}}
\label{8L5}
}
\subfigure[$\gamma=\dfrac{\sin 2\pi/8}{\sin 4\pi/8}, \theta=0$]{
{
\begin{tabular}{|c|c|c|c|c|c|c|c|}
\hline 4 & 0 & 6 & 3 & 1 & 5 & 7 & 2\\ 
\hline 2 & 7 & 0 & 1 & 3 & 6 & 5 & 4\\
\hline 0 & 4 & 7 & 2 & 5 & 1 & 6 & 3\\  
\hline 3 & 6 & 4 & 0 & 2 & 7 & 1 & 5\\ 
\hline 5 & 1 & 2 & 6 & 0 & 4 & 3 & 7\\ 
\hline 7 & 2 & 1 & 5 & 6 & 3 & 4 & 0\\
\hline 1 & 5 & 3 & 7 & 4 & 0 & 2 & 6\\ 
\hline 6 & 3 & 5 & 4 & 7 & 2 & 0 & 1\\
\hline 
\end{tabular} }
\label{8L6}
}
\subfigure[$\gamma=\dfrac{\sin 2\pi/8}{\sin4\pi/8}, \theta=\pi/2$]{
{
\begin{tabular}{|c|c|c|c|c|c|c|c|}
\hline 6 & 3 & 1 & 5 & 7 & 2 & 4 & 0\\ 
\hline 0 & 1 & 3 & 6 & 5 & 4 & 2 & 7\\
\hline 7 & 2 & 5 & 1 & 6 & 3 & 0 & 4\\  
\hline 4 & 0 & 2 & 7 & 1 & 5 & 3 & 6\\ 
\hline 2 & 6 & 0 & 4 & 3 & 7 & 5 & 1\\ 
\hline 1 & 5 & 6 & 3 & 4 & 0 & 7 & 2\\
\hline 3 & 7 & 4 & 0 & 2 & 6 & 1 & 5\\ 
\hline 5 & 4 & 7 & 2 & 0 & 1 & 6 & 3\\
\hline 
\end{tabular}}
\label{8L7}
}
\subfigure[$\gamma=\dfrac{\sin 4\pi/8}{\sin 2\pi/8}, \theta=-\pi/2$]{
{
\begin{tabular}{|c|c|c|c|c|c|c|c|}
\hline 6 & 0 & 7 & 4 & 2 & 1 & 3 & 5\\ 
\hline 3 & 1 & 2 & 0 & 6 & 5 & 7 & 4\\
\hline 1 & 3 & 5 & 2 & 0 & 6 & 4 & 7\\  
\hline 5 & 6 & 1 & 7 & 4 & 3 & 0 & 2\\ 
\hline 7 & 5 & 6 & 1 & 3 & 4 & 2 & 0\\ 
\hline 2 & 4 & 3 & 5 & 7 & 0 & 6 & 1\\
\hline 4 & 2 & 0 & 3 & 5 & 7 & 1 & 6\\ 
\hline 0 & 7 & 4 & 6 & 1 & 2 & 5 & 3\\
\hline 
\end{tabular} }
\label{8L8}
} 
\caption[]{Latin Squares Corresponding to Different singular fade states }
\end{figure}
\begin{figure}[h]
\centering
\subfigure[$\gamma=\dfrac{\sin \pi/8}{\sin 4\pi/8}, \theta=\pi/8$]{
{
\begin{tabular}{|c|c|c|c|c|c|c|c|}
\hline 4 & 3 & 0 & 6 & 2 & 1 & 7 & 5 \\ 
\hline 1 & 5 & 3 & 2 & 4 & 7 & 0 & 6 \\
\hline 5 & 6 & 7 & 1 & 3 & 0 & 4 & 2 \\  
\hline 3 & 2 & 6 & 5 & 0 & 4 & 1 & 7 \\ 
\hline 0 & 4 & 2 & 7 & 6 & 3 & 5 & 1 \\ 
\hline 6 & 7 & 5 & 0 & 1 & 2 & 3 & 4 \\
\hline 7 & 0 & 1 & 4 & 5 & 6 & 2 & 3 \\ 
\hline 2 & 1 & 4 & 3 & 7 & 5 & 6 & 0 \\
\hline 
\end{tabular}}
\label{8L9}
}
\subfigure[$\gamma=\dfrac{\sin \pi/8}{\sin 2\pi/8}, \theta=\pi/8$]{
{
\begin{tabular}{|c|c|c|c|c|c|c|c|}
\hline 6 & 0 & 7 & 1 & 2 & 4 & 3 & 5\\ 
\hline 7 & 3 & 4 & 2 & 5 & 1 & 6 & 0\\
\hline 4 & 2 & 0 & 7 & 1 & 6 & 5 & 3\\  
\hline 0 & 7 & 1 & 4 & 3 & 5 & 2 & 6\\ 
\hline 2 & 4 & 3 & 5 & 6 & 0 & 7 & 1\\ 
\hline 5 & 1 & 6 & 0 & 7 & 3 & 4 & 2\\
\hline 1 & 6 & 5 & 3 & 4 & 2 & 0 & 7\\ 
\hline 3 & 5 & 2 & 6 & 0 & 7 & 1 & 4\\
\hline 
\end{tabular}  }
\label{8L10}
}
\subfigure[$\gamma=\dfrac{\sin 2\pi/8}{\sin 3\pi/8}, \theta=\pi/8$]{
{
 \begin{tabular}{|c|c|c|c|c|c|c|c|}
\hline 7 & 5 & 0 & 1 & 2 & 3 & 4 & 6 \\ 
\hline 0 & 1 & 3 & 4 & 5 & 6 & 7 & 2 \\
\hline 6 & 4 & 5 & 7 & 3 & 2 & 1 & 0 \\  
\hline 4 & 0 & 2 & 3 & 1 & 7 & 6 & 5 \\ 
\hline 2 & 3 & 4 & 6 & 7 & 5 & 0 & 1 \\ 
\hline 5 & 6 & 7 & 2 & 0 & 1 & 3 & 4 \\
\hline 3 & 2 & 1 & 0 & 6 & 4 & 5 & 7 \\ 
\hline 1 & 7 & 6 & 5 & 4 & 0 & 2 & 3 \\
\hline 
\end{tabular} }
\label{8L11}
}
\subfigure[$\gamma=\dfrac{\sin 3\pi/8}{\sin 4\pi/8}, \theta=\pi/8$]{
{
\begin{tabular}{|c|c|c|c|c|c|c|c|}
\hline 0 & 1 & 2 & 3 & 6 & 5 & 4 & 7\\ 
\hline 6 & 5 & 0 & 7 & 3 & 4 & 2 & 1\\
\hline 5 & 6 & 7 & 0 & 4 & 3 & 1 & 2\\  
\hline 1 & 7 & 3 & 6 & 0 & 2 & 5 & 4\\ 
\hline 3 & 0 & 1 & 4 & 2 & 7 & 6 & 5\\ 
\hline 4 & 3 & 6 & 2 & 5 & 1 & 7 & 0\\
\hline 2 & 4 & 5 & 1 & 7 & 0 & 3 & 6\\ 
\hline 7 & 2 & 4 & 5 & 1 & 6 & 0 & 3\\
\hline 
\end{tabular} }
\label{8L12}
} 
\caption[]{Latin Squares Corresponding to Different singular fade states }
\end{figure}

\begin{table*}
\centering
 \caption{Clusterings Obtained for different singular fade states on each circle (for $\gamma \leq 1$) when the end nodes use 8-PSK constellations}
 \label{table2}
 \begin{tabular}{|c|c|c|}

\hline Sl.No & Singular fade states  & Cluster\\ 
\hline &    &$\{\{(0,1)(1,0)(2,3)(3,2)(5,4)(4,5)(6,7)(7,6)\},\{(0,2)(1,3)(2,0)(3,1)(4,6)(6,4)(5,7)(7,5)\}$ \\ 1&$\gamma=1,\theta=0$ & $\{(0,3)(3,0)(1,2)(2,1)(4,7)(5,6)(6,5)(7,4)\},\{(0,4)(4,0)(1,5)(5,1)(3,7)(2,6)(6,2)(7,3)\}$ \\ & & $\{(0,5)(1,4)(2,7)(3,6)(4,1)(5,0)(6,3)(7,2)\},\{(0,6)(6,0)(1,7)(7,1)(4,2)(5,3)(3,5)(2,4)\}$ \\ & & $\{(0,7)(7,0)(1,6)(6,1)(4,3)(5,2)(2,5)(3,4)\},\{(0,0)(1,1)(2,2)(3,3)(4,4)(5,5)(6,6),(7,7)\}\}$\\ 
\hline &   &$\{\{(0,0)(1,5)(2,2)(3,7)(4,4)(5,1)(6,6)(7,3)\},\{(0,1)(1,4)(2,7)(3,2)(4,5)(5,0)(6,6)(7,3)\}$ \\ & &$\{(0,2)(3,1)(1,7)(2,4)(4,6)(5,3)(6,0)(7,5)\},\{(0,3)(1,6)(2,1)(3,4)(4,7)(5,2)(6,5)(7,0)\}$ \\ 2&$\gamma=\dfrac{\sin\pi/8}{\sin3\pi/8},\theta=0$ &$\{(0,4)(1,1)(2,6)(3,3)(4,0)(5,5)(6,2)(7,7)\},\{(0,5)(3,6)(1,0)(2,3)(4,1)(5,4)(6,7)(7,2)\}$ \\ & &$\{(0,6)(3,5)(1,3)(2,0)(4,2)(5,3)(6,4)(7,1)$\},$\{(0,7)(3,0)(1,2)(2,5)(4,3)(5,6)(6,1)(7,4)\}\}$\\ 
\hline & &$\{\{(0,0)(1,7)(2,1)(3,2)(4,5)(5,6)(6,4)(7,3)\},\{(0,1)(1,2)(2,0)(3,3)(4,4)(5,7)(6,5)(7,6)\}$ \\ & &$\{(0,2)(3,1)(1,5)(2,6)(4,3)(5,4)(6,7)(7,0)\},\{(0,3)(1,4)(2,7)(3,0)(4,6)(5,5)(6,2)(7,1)\}$ \\3&$\gamma=\dfrac{\sin2\pi/8}{\sin4\pi/8},\theta=0$                                           &$\{(0,4)(1,3)(2,5)(3,6)(4,1)(5,2)(6,0)(7,7)\},\{(0,5)(1,6)(3,7)(2,4)(4,0)(5,3)(6,1)(7,2)\}$ \\ & &$\{(0,6)(3,5)(1,1)(2,2)(4,7)(5,0)(6,3)(7,4)$\},$\{(0,7)(1,0)(2,3)(3,4)(4,2)(5,1)(6,6)(7,5)\}\}$\\ 
\hline & &$\{\{(0,0)(1,4)(2,6)(3,5)(4,1)(5,7)(6,3)(7,2)\},\{(0,1)(1,2)(2,4)(3,0)(4,5)(5,6)(6,7)(7,3)\}$\\ & &$\{(0,2)(1,6)(2,5)(3,4)(4,0)(5,3)(6,1)(7,7)\},\{(0,3)(1,7)(2,1)(3,2)(4,4)(5,0)(6,5)(7,6)\}$ \\ 4&$\gamma=\dfrac{\sin\pi/8}{\sin4\pi/8},\theta=\pi/8$                                           &$\{(0,4)(1,3)(2,7)(3,1)(4,2)(5,5)(6,6)(7,0)\},\{(0,5)(1,0)(2,3)(3,6)(4,7)(5,4)(6,2)(7,1)\}$ \\ & &$\{(0,6)(1,5)(2,2)(3,7)(4,3)(5,1)(6,0)(7,4)\},\{(0,7)(1,1)(2,0)(3,3)(4,6)(5,2)(6,4)(7,5)\}\}$\\ 
\hline &                                           &$\{\{(0,0)(1,6)(2,5)(3,7)(4,4)(5,2)(6,1)(7,3)\},\{(0,1)(1,7)(2,2)(3,0)(4,5)(5,3)(6,6)(7,4)\}$ \\& &$\{(0,2)(1,0)(2,3)(3,1)(4,6)(5,4)(6,7)(7,5)\},\{(0,3)(1,5)(2,4)(3,2)(4,7)(5,1)(6,0)(7,6)\}$ \\ 5&$\gamma=\dfrac{\sin\pi/8}{\sin2\pi/8},\theta=\pi/8$ &$\{(0,4)(1,3)(2,1)(3,6)(4,0)(5,7)(6,5)(7,2)\},\{(0,5)(1,2)(2,0)(3,3)(4,1)(5,6)(6,4)(7,7)\}$ \\ & &$\{(0,6)(1,1)(2,7)(3,4)(4,2)(5,5)(6,3)(7,0)\},\{(0,7)(1,4)(2,6)(3,5)(4,3)(5,0)(6,2)(7,1)\}\}$\\ 
\hline &                                           &$\{\{(0,0)(1,6)(2,3)(3,5)(4,4)(5,2)(6,7)(7,1)\},\{(0,1)(1,4)(2,2)(3,7)(4,5)(5,0)(6,6)(7,3)\}$ \\ & &$\{(0,2)(1,0)(2,7)(3,1)(4,6)(5,4)(6,3)(7,5)\},\{(0,3)(1,1)(2,6)(3,4)(4,7)(5,5)(6,2)(7,0)\}$ \\ 6&$\gamma=\dfrac{\sin2\pi/8}{\sin3\pi/8},\theta=\pi/8$ &$\{(0,4)(1,7)(2,5)(3,2)(4,0)(5,3)(6,1)(7,6)\},\{(0,5)(1,2)(2,4)(3,3)(4,1)(5,6)(6,0)(7,7)\}$ \\ & &$\{(0,6)(1,3)(2,1)(3,0)(4,2)(5,7)(6,5)(7,4)\},\{(0,7)(1,5)(2,0)(3,6)(4,3)(5,1)(6,4)(7,2)\}\}$\\ 
\hline & &$\{\{(0,0)(1,2)(2,3)(3,4)(4,1)(5,7)(6,5)(7,6)\},\{(0,1)(1,7)(2,6)(3,0)(4,2)(5,5)(6,3)(7,4)\}$\\ & &$\{(0,2)(1,6)(2,7)(3,5)(4,4)(5,3)(6,0)(7,1)\},\{(0,3)(1,4)(2,5)(3,2)(4,0)(5,1)(6,6)(7,7)\}$ \\ 7&$\gamma=\dfrac{\sin3\pi/8}{\sin4\pi/8},\theta=\pi/8$                                           &$\{(0,4)(1,0)(2,1)(3,3)(4,6)(5,2)(6,7)(7,5)\},\{(0,5)(1,1)(2,0)(3,6)(4,7)(5,4)(6,2)(7,3)\}$ \\ & &$\{(0,6)(1,5)(2,4)(3,7)(4,3)(5,0)(6,7)(7,2)\},\{(0,7)(1,3)(2,2)(3,1)(4,5)(5,6)(6,4)(7,0)\}\}$\\ 

\hline 
\end{tabular} 
\vspace{-.1 cm}
\end{table*}
 The Latin Square shown in figure.\ref{8L3} removes this singular fade state. Further by column shifting all the singular fade states with $\gamma=0.414$ are removed. Consider  $\gamma=\dfrac{\sin \pi/8}{\sin 3\pi/8} =0.414, \theta=\pi/4$. The singularity constraints are 

\begin{center}
$\{(0,1)(1,6)\}$, $\{(0,2)(1,5)\}$, $\{(0,4)(7,1)\}$, $\{(0,5)(7,0)\}$\\*
$\{(1,2)(2,7)\}$, $\{(1,3)(2,6)\}$, $\{(2,3)(3,0)\}$, $\{(2,4)(3,7)\}$\\*
$\{(3,4)(4,1)\}$, $\{(3,5)(4,0)\}$, $\{(4,5)(5,2)\}$, $\{(4,6)(5,1)\}$\\*
$\{(5,7)(6,2)\}$, $\{(5,6)(6,3)\}$, $\{(6,7)(7,4)\}$, $\{(6,0)(7,3)\}$\\*
\end{center}
The Latin Square which removes the above singular fade state is shown in figure.\ref{8L4}. Latin Squares which removes the singular fade states with $\gamma=\dfrac{\sin 3\pi/8}{\sin \pi/8}$ are obtained by taking transposes of the Latin Squares which removes the singular fade states with $\gamma=\dfrac{\sin \pi/8}{\sin 3\pi/8}$. For example, consider $\gamma=\dfrac{\sin 3\pi/8}{\sin \pi/8}, \theta=-\pi/4$. The singularity constraints are 
\begin{center}
$\{(1,0)(6,1)\}$, $\{(2,0)(5,1)\}$, $\{(4,0)(1,7)\}$, $\{(5,0)(0,7)\}$\\*
$\{(2,1)(7,2)\}$, $\{(3,1)(6,2)\}$, $\{(3,2)(0,3)\}$, $\{(4,2)(7,3)\}$\\*
$\{(4,3)(1,4)\}$, $\{(5,3)(0,4)\}$, $\{(5,4)(2,5)\}$, $\{(6,4)(1,5)\}$\\*
$\{(7,5)(2,6)\}$, $\{(6,5)(3,6)\}$, $\{(7,6)(4,7)\}$, $\{(0,6)(3,7)\}$\\*
\end{center}

The Latin Square to remove above singular fading state (figure.\ref{8L5}) is obtained by taking transpose of the Latin Square shown in figure.\ref{8L4}. Then by column shifting all the singular fading states with $\gamma=\dfrac{\sin 3\pi/8}{\sin \pi/8}$ is removed.

Next Consider $\gamma=\sin \pi/4, \theta=0$. The singularity constraints are 
\begin{center}
$\{(0,3)(2,7)\}$, $\{(0,5)(6,1)\}$, $\{(1,4)(3,0)\}$, $\{(1,6)(7,2)\}$\\*
$\{(2,5)(4,1)\}$, $\{(3,6)(5,2)\}$, $\{(4,7)(6,3)\}$, $\{(5,0)(7,4)\}$\\*
\end{center}

The Latin Square shown in figure.\ref{8L6} removes this singular fade state. Then by shifting columns all the singular fade states with $\gamma=\sin \pi/4$ is removed. By taking transposes all the singular fade states with $\gamma=\dfrac{1}{\sin \pi/4}$ is removed.

For example, the Latin Square for singular fade state $\gamma=\sin \pi/4, \theta=\pi/2$ is obtained by two column shifting from Latin Square for $\gamma=\sin \pi/4, \theta=0$. The singularity constraints for $\gamma=\sin \pi/4, \theta=\pi/2$ are,
\begin{center}
$\{(0,1)(2,5)\}$, $\{(0,3)(6,7)\}$, $\{(1,2)(3,6)\}$, $\{(1,4)(7,0)\}$\\*
$\{(2,3)(4,7)\}$, $\{(3,4)(5,0)\}$, $\{(4,5)(6,1)\}$, $\{(5,6)(7,2)\}$\\*
\end{center}
The Latin Square to remove singular fade state  $\gamma=\sin \pi/4, \theta=\pi/2$ is shown in Table.\ref{8L7}.

Now consider the singular fade state $\gamma=\dfrac{1}{\sin \pi/4}, \theta=-\pi/2$. The singularity constraints for this singular fade state are  
\begin{center}
$\{(1,0)(5,2)\}$, $\{(3,0)(7,6)\}$, $\{(2,1)(6,3)\}$, $\{(4,1)(0,7)\}$\\*
$\{(3,2)(7,4)\}$, $\{(4,3)(0,5)\}$, $\{(5,4)(1,6)\}$, $\{(6,5)(2,7)\}$\\*
\end{center}
From the Latin Square to remove singular fade state  $\gamma=\sin \pi/4, \theta=\pi/2$ the Latin Square to remove singular fade state  $\gamma=\dfrac{1}{\sin \pi/4}, \theta=-\pi/2$ (figure.\ref{8L8}) is obtained by taking transpose.

Till now we have removed 40 singular fade states lying on the angles $\theta=2m\pi/8$, where $m=0,1,2\cdots7$. Now $\theta=(2m+1)\pi/8$, where $m=0,1,2\cdots7$. is considered.

The singularity constraints for the singular fade state $\gamma=\dfrac{\sin\pi/8}{\sin4\pi/8}$ and $\theta=\pi/8$ are
\begin{center}
$\{(0,2)(1,6)\}, \{(0,5)(7,1)\}, \{(1,3)(2,7)\}, \{(2,4)(3,0)\}$\\*
$\{(3,5)(4,1)\}, \{(4,6)(5,2)\}, \{(5,7)(6,3)\}, \{(6,0)(7,4)\}$\\*
\end{center}
 The Latin Square to remove this singular fade state is given in figure.\ref{8L9}. By shifting the columns of this Latin Square all other Latin Squares to remove singular fade states with $\gamma=\dfrac{\sin\pi/8}{\sin4\pi/8}$ is obtained. The transposes of these Latin Squares remove singular fade states with $\gamma=\dfrac{\sin4\pi/8}{\sin\pi/8}$.
 
The singularity constraints for the singular fade state $\gamma=\dfrac{\sin\pi/8}{\sin2\pi/8}$ and $\theta=\pi/8$ are
\begin{center}
$\{(0,1)(1,7)\}, \{(0,3)(1,5)\}, \{(0,4)(7,2)\}, \{(0,6)(7,0)\}$\\*
$\{(1,2)(2,0)\}, \{(1,4)(2,6)\}, \{(2,3)(3,1)\}, \{(2,5)(3,7)\}$\\*
$\{(3,4)(4,2)\}, \{(3,6)(4,0)\}, \{(4,5)(5,3)\}, \{(4,7)(5,1)\}$\\*
$\{(5,0)(6,2)\}, \{(5,6)(6,4)\}, \{(6,1)(7,3)\}, \{(6,7)(7,5)\}$\\*
\end{center}
 The Latin Square to remove this singular fade state is given in figure.\ref{8L10}. By shifting the columns of this Latin Square all other Latin Squares to remove singular fade states with $\gamma=\dfrac{\sin\pi/8}{\sin2\pi/8}$ is obtained. The transposes of these Latin Squares remove singular fade states with $\gamma=\dfrac{\sin2\pi/8}{\sin\pi/8}$.

The singularity constraints for the singular fade state $\gamma=\dfrac{\sin2\pi/8}{\sin3\pi/8}$ and $\theta=\pi/8$ are
\begin{center}
$\{(0,2)(2,7)\}, \{(0,3)(2,6)\}, \{(0,4)(6,1)\}, \{(0,5)(6,0)\}$\\*
$\{(1,3)(3,0)\}, \{(1,4)(3,7)\}, \{(1,5)(7,2)\}, \{(1,6)(7,1)\}$\\*
$\{(2,4)(4,1)\}, \{(2,5)(4,0)\}, \{(3,5)(5,2)\}, \{(3,6)(5,1)\}$\\*
$\{(4,6)(6,3)\}, \{(4,7)(6,2)\}, \{(5,0)(7,3)\}, \{(5,7)(7,4)\}$\\*
\end{center}
 The Latin Square to remove this singular fade state is given in figure.\ref{8L11}. By shifting the columns of this Latin Square all other Latin Squares to remove singular fade states with $\gamma=\dfrac{\sin2\pi/8}{\sin3\pi/8}$ is obtained. The transposes of these Latin Squares remove singular fade states with $\gamma=\dfrac{\sin3\pi/8}{\sin2\pi/8}$.

 The singularity constraints for the singular fade state $\gamma=\dfrac{\sin3\pi/8}{\sin4\pi/8}$ and $\theta=\pi/8$ are
\begin{center}
$\{(0,3)(3,7)\}, \{(0,4)(5,0)\}, \{(1,4)(4,0)\}, \{(1,5)(6,1)\}$\\*
$\{(2,5)(5,1)\}, \{(2,6)(7,2)\}, \{(3,6)(6,2)\}, \{(4,7)(7,3)\}$\\*
\end{center}
 The Latin Square to remove this singular fade state is given in figure.\ref{8L12}. By shifting the columns of this Latin Square all other Latin Squares to remove singular fade states with $\gamma=\dfrac{\sin3\pi/8}{\sin4\pi/8}$ is obtained. The transposes of these Latin Squares remove singular fade states with $\gamma=\dfrac{\sin4\pi/8}{\sin3\pi/8}$.
  From above discussions it is clear that by obtaining one Latin Square for each value of $\gamma \leq 1$ the other Latin Squares can be obtained by column permutation or by taking transpose. The clustering for the seven $\gamma$ values when end nodes use 8-PSK is explicitly given in \ref{table2}. All the clusterings are having  8 clusters only. Thus all the singular fade states are removed with 8 symbols only. That is, relay can use 8-PSK constellation.
 
\section{Singularity Removal Constraints FOR $M$-PSK SIGNAL SET}

It was shown in Section III A that the singularity removal constraints can be equivalently viewed as a Constrained Partial Latin Square (CPLS). In the following subsection, the structure of the CPLS for $M$-PSK signal set is obtained. 

\subsection{The Structure of Constrained Partial Latin Square for $M$-PSK Signal Set}
In this subsection, the structure of the CPLS is shown to be the one referred to as $t$-plex which is defined as follows:

\begin{definition}
A $n \times n$ Partially Filled Latin Square (PFLS) is said to be a $t$-plex of order $n$ with $q$ symbols, if exactly $t$ cells are filled in each row and in each column, with each one of the $q$ symbols filling exactly $nt/q$ cells in total.  A $n \times n$ PFLS which is a $t$-plex with $n$ symbols is simply referred to as a $t$-plex of order $n$.
\end{definition}

\begin{figure}
\centering
\includegraphics[totalheight=1in,width=1in]{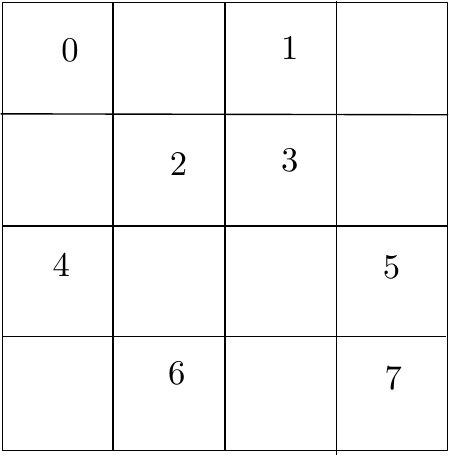}
\caption{Example showing a 2-plex of order 4 with 8 symbols}
\label{fig:k_plex_ex2}
\end{figure}

\begin{figure}
\centering
\includegraphics[totalheight=1in,width=1in]{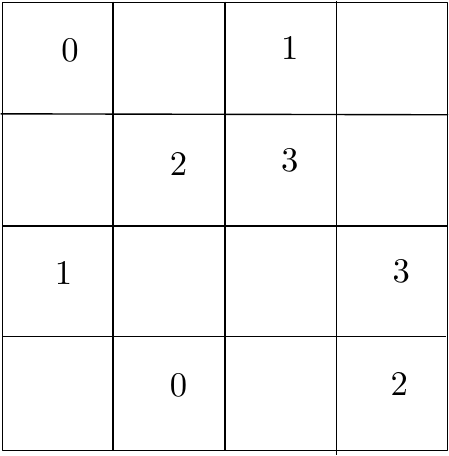}
\caption{Example showing a 2-plex of order 4}
\label{fig:k_plex_ex1}
\end{figure}

For example, Fig. \ref{fig:k_plex_ex2} shows a 2-plex of order 4 with 8 symbols and Fig. \ref{fig:k_plex_ex1} shows a 2-plex of order 4. Note that a $t$-plex of order $n$ need not be always completable to a filled Latin Square with $n$ symbols. For example, the 2-plex of order 4 shown in Fig. \ref{fig:k_plex_ex1} cannot be completed using 4 symbols, since the cell corresponding to the zeroth row and third column cannot be filled with any of the four symbols $0,1,2,3$.

%
%

The CPLS for the singular fade states which lie on the unit circle and the corresponding Latin Squares which remove them, were obtained in Section III A. The following lemma gives the structure of the CPLS for those singular fade states which lie on circles $\frac{\sin(\frac{\pi k}{M})}{\sin(\frac{\pi l}{M})}, k \neq l$, for the case when both $k$ and $l$ are odd or both $k$ and $l$ are even.
\begin{lemma}
\label{lemma_sing_const1}
For a singular fade state $\frac{\sin\left(\frac{\pi k}{M}\right)}{\sin\left(\frac{\pi l}{M}\right)}e^{j 2\pi m/M}$, where both $k$ and $l$ are odd or both are even, the singularity removal constraints are given by, 

{\footnotesize
\begin{align}
\nonumber
&\left\lbrace\left(i,i-m-\frac{M}{2}-\frac{\left(k-l\right)}{2}\right),\left(i-k,i-m+\frac{M}{2}-\frac{\left(k+l\right)}{2}\right)\right\rbrace,\\
\nonumber 
&\left\lbrace\left(i,i-m-\frac{\left(k+l\right)}{2}\right),\left(i-k,i-m-\frac{\left(k-l\right)}{2}\right)\right\rbrace,
\end{align}
}where $0 \leq i \leq M-1$.
\begin{proof}
For $\lbrace (k_1,l_1), (k_2,l_2)\rbrace$ to be a singularity constraint corresponding to the singular fade state $\frac{\sin\left(\frac{\pi k}{M}\right)}{\sin\left(\frac{\pi l}{M}\right)}$,

{\footnotesize
\begin{align}
\nonumber
-\frac{e^{ \frac{jk_1\pi}{M}}-e^{\frac{jk_2\pi}{M}}}{e^{\frac{jl_1\pi}{M}}-e^{\frac{jl_2\pi}{M}}}=\frac{\sin\left(\frac{\pi k}{M}\right)}{\sin\left(\frac{\pi l}{M}\right)}e^{j 2\pi m/M},\textrm{ i.e.,}
\end{align}
}
{\footnotesize
\begin{align}
\label{eqn_equate}
\frac{\sin(\dfrac{\pi(k_1-k_2)}{M})}{\sin(\dfrac{\pi(l_2-l_1)}{M})}e^{\frac{j(k_1+k_2-l_1-l_2)\pi}{M}}=\frac{\sin\left(\frac{\pi k}{M}\right)}{\sin\left(\frac{\pi l}{M}\right)}e^{j 2\pi m/M}.
\end{align}
}

Equating the amplitude and phase terms in \eqref{eqn_equate}, we get, $$k_1+k_2=l_1+l_2,k_1-k_2=k,l_2-l_1=l$$ or, $$k_1+k_2=l_1+l_2+2m,k_1-k_2=k,l_2-l_1=M-l.$$
Note that by Lemma \ref{lemma_trig}, separately equating the numerator and denominator of the amplitude on both the sides of \eqref{eqn_equate} is allowed.

Assuming $k_1=i$, the two solutions given in the statement of the lemma are obtained. 
\end{proof}  
\end{lemma}

\begin{figure}
\centering
\includegraphics[totalheight=2in,width=2in]{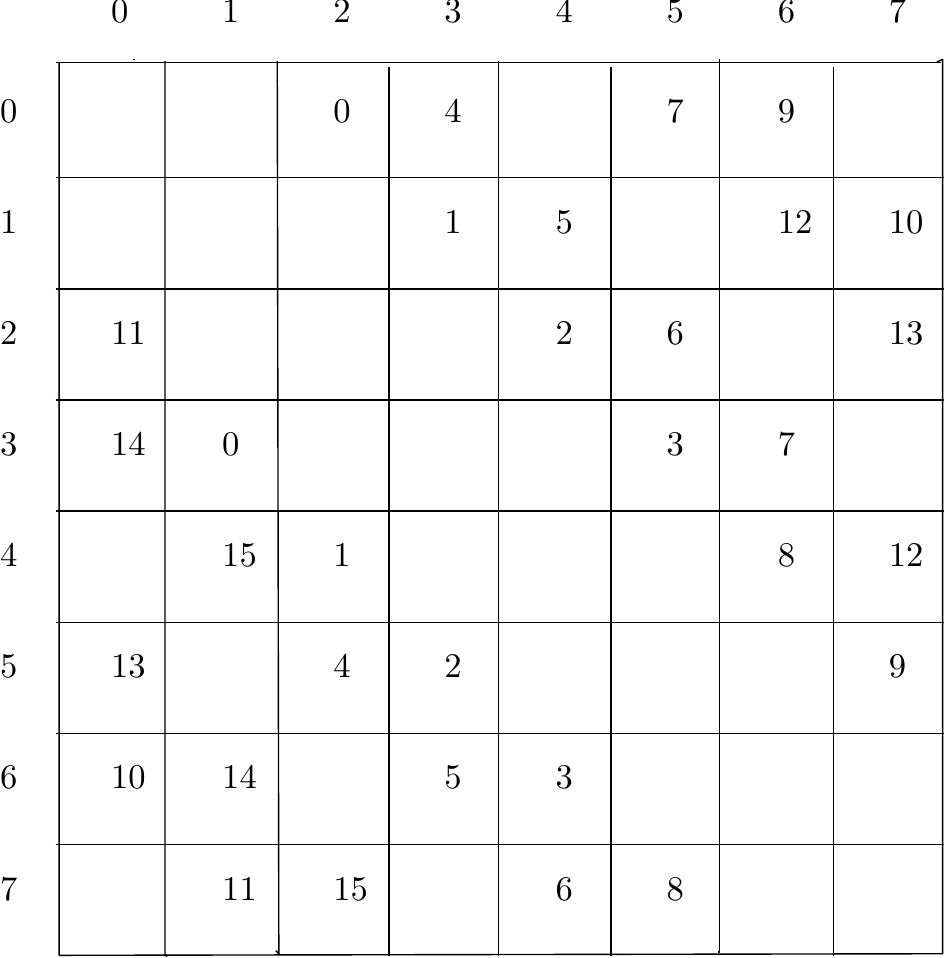}
\caption{CPLS corresponding to the singular fade state $\frac{\sin(\frac{3 \pi}{8})}{\sin(\frac{\pi}{8})}$}
\label{fig:PFLS_3_1}
\vspace{-.9 cm}
\end{figure}

For example, the CPLS corresponding to the singular fade state $\frac{\sin(\frac{3 \pi}{8})}{\sin(\frac{\pi}{8})}$, is shown in Fig. \ref{fig:PFLS_3_1}.

The following lemma gives the structure of the CPLS for those singular fade states which lie on circles $\frac{\sin(\frac{\pi k}{M})}{\sin(\frac{\pi k}{M})}, k \neq l$, for the case when only one among  $k$ and $l$ is odd.
\begin{lemma}
\label{lemma_sing_const2}
For a singular fade state $\frac{\sin\left(\frac{\pi k}{M}\right)}{\sin\left(\frac{\pi l}{M}\right)}e^{j \frac{(2m+1)\pi}{M}}$, where only one among $k$ and $l$ is odd, the singularity removal constraints are given by,

{\scriptsize
\begin{align*} 
&\left\lbrace\left(i,i-m-\frac{M}{2}-\frac{\left(k+1-l\right)}{2}\right),\left(i-k,i-m+\frac{M}{2}-\frac{\left(k+1+l\right)}{2}\right)\right\rbrace,\\ 
&\left\lbrace\left(i,i-m-\frac{\left(k+1+l\right)}{2}\right),\left(i-k,i-m-\frac{\left(k+1-l\right)}{2}\right)\right\rbrace,
\end{align*}
}where $0 \leq i \leq M-1$.
\begin{proof}
The proof is similar to the proof of Lemma \ref{lemma_sing_const1} and is omitted.
\end{proof}  
\end{lemma}

The structure of the CPLS for $M$-PSK signal set, is shown to be a 4-plex of order $M$ with $2M$ symbols or a 2-plex of order $M$, as stated in the following lemma. 
\begin{lemma}
The CPLS obtained after filling in the singularity removal constraints corresponding to the singular fade state which lies on the circle with radius $\frac{\sin(\frac{\pi k}{M})}{\sin(\frac{\pi l}{M})}$, $k \neq l$, forms a 4-plex of order $M$ with $2M$ symbols if $k \neq M/2,l \neq M/2$ and forms a 2-plex of order $M$ otherwise.
\begin{proof}
\noindent 
{\it Case (i) $k \neq M/2$, $l \neq M/2$:}
From Lemma \ref{lemma_sing_const1} and Lemma \ref{lemma_sing_const2}, each value of $i$, $0 \leq i \leq M-1$, gives rise to two distinct constraints, which fills two cells in the $i$-th row. Similarly, the two constraints corresponding to $i'=i+k$ also fills two cells in the $i$-th row. Hence, in each row, 4 cells are filled. By a similar argument, 4 cells are filled in each column. Since the same symbol fills in two of the cells, a total of $4M/2=2M$ symbols are used in the partial filling of Latin Square. As a result, a $4$-plex with $2M$ symbols is obtained. 

\noindent
{\it Case (ii) $k = M/2$ or  $l = M/2$:} 
When $k = M/2$ or  $l = M/2$, from Lemma \ref{lemma_sing_const1} and Lemma \ref{lemma_sing_const2}, for each value of $i$, $0 \leq i \leq M-1$, it can be verified that the two singularity constraints are not distinct.  Only 2 cells are filled in each row and each column. As a result, a $2$-plex is obtained. 
\end{proof}
\end{lemma}

For example, the PFLS corresponding to the singular fade state $\frac{\sin(\frac{3 \pi}{8})}{\sin(\frac{\pi}{8})}$, shown in Fig. \ref{fig:PFLS_3_1}, is a 4-plex of order 8 with 16 symbols.

If the same symbol is filled in those cells which belong to two different singularity constraints, the singularity removal constraints are said to be combined. It is shown in the following lemma that by combining the singularity constraints, the CPLS which is a $4$-plex of order $M$ with $2M$ symbols can be transformed into a $4$-plex of order $M$.

\begin{lemma}
\label{lemma_PFLS}
For $k\neq M/2,l \neq M/2$, the CPLS corresponding to the singular fade state which lies on the circle $\frac{\sin(\frac{\pi k}{M})}{\sin(\frac{\pi k}{M})}, k \neq l$, obtained after filling in the singularity constraints, can be transformed in to a $4$-plex by appropriately combining the singularity constraints.
\begin{proof}
Since the elements of both the sets $\lbrace i,i-k,i+M/2,i-k+M/2 \rbrace$ and $\lbrace i-M/2-(k-l)/2,i+M/2-(k+l)/2,i-(k-l)/2,i-(k+l)/2 \rbrace$ are distinct, the following constraints in Lemma \ref{lemma_sing_const1}, can be combined:

{\footnotesize
\begin{align*}
&\left\lbrace\left(i,i-\frac{M}{2}-\frac{\left(k-l\right)}{2}\right),\left(i-k,i+\frac{M}{2}-\frac{\left(k+l\right)}{2}\right)\right\rbrace,\\
&\left\lbrace\left(i+\frac{M}{2},i-\frac{\left(k-l\right)}{2}\right),\left(i-k+\frac{M}{2},i-\frac{\left(k+l\right)}{2}\right)\right\rbrace.
\end{align*}
} 

Similarly, the following constraints can also be combined:

{\footnotesize
\begin{align*}
&\left\lbrace\left(i,i-\frac{\left(k+l\right)}{2}\right),\left(i-k,i-\frac{\left(k-l\right)}{2}\right)\right\rbrace,\\
&\left\lbrace\left(i+\frac{M}{2},i+\frac{M}{2}-\frac{\left(k+l\right)}{2}\right),\left(i-k+\frac{M}{2},i+\frac{M}{2}-\frac{\left(k-l\right)}{2}\right)\right\rbrace.
\end{align*}
}
In this way, a 4-plex is obtained from the PFLS for the case when $k,l \neq \frac{M}{2}$ and both $k$ and $l$ are even or both are odd. Following, exactly the same procedure, the lemma can be proved for the case when one among $k$ and $l$ is odd and the other is even. 
\end{proof}
\end{lemma}

\begin{figure}
\centering
\includegraphics[totalheight=2in,width=2in]{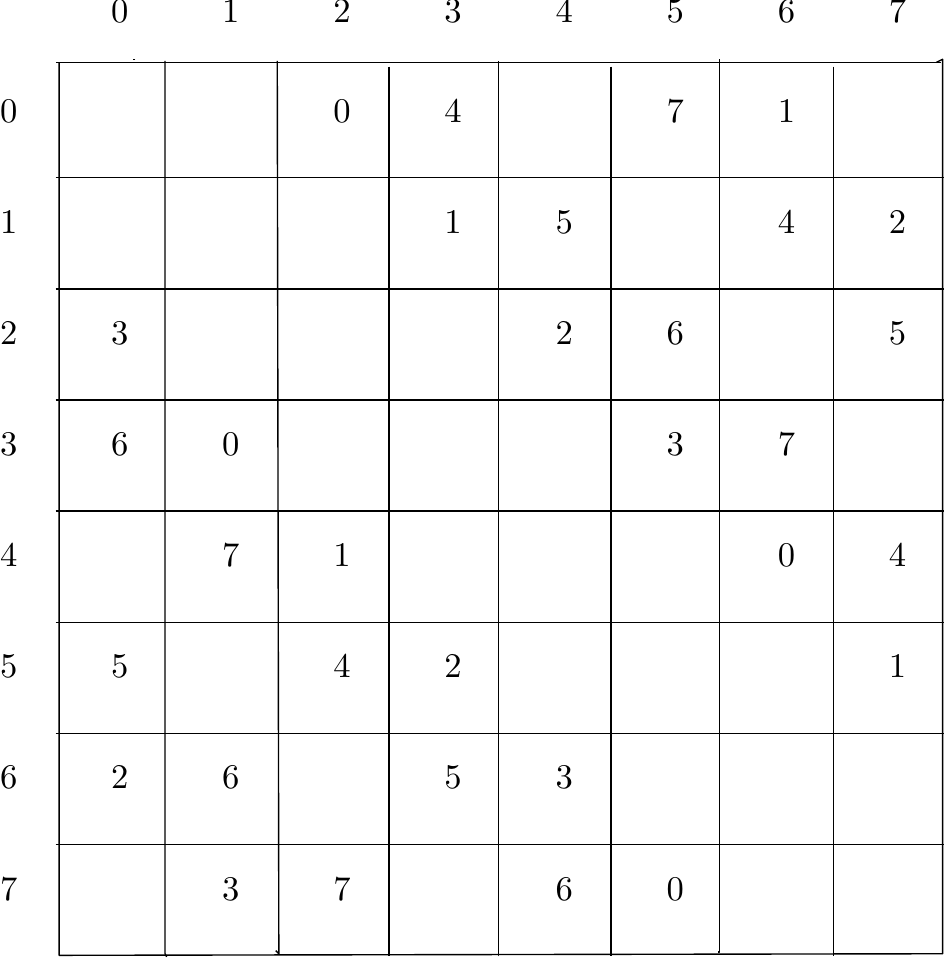}
\caption{4-plex of order 8 corresponding to the singular fade state $\frac{\sin(\frac{3 \pi}{8})}{\sin(\frac{\pi}{8})}$}
\label{fig:PFLS_3_1_combined}
\vspace{-.8 cm}
\end{figure}

 For example, the CPLS corresponding to the singular fade state $\frac{\sin(\frac{3 \pi}{8})}{\sin(\frac{\pi}{8})}$, shown in Fig. \ref{fig:PFLS_3_1}, can be transformed into a 4-plex of order 8, as shown in Fig. \ref{fig:PFLS_3_1_combined}. 
 
 The completability of a $t$-plex of order $M$ with $M$ symbols is discussed in the following subsection.
 \subsection{Completability of a $t$-plex of order $M$ with $M$ symbols} 
The completion of the 2-plexes and the 4-plexes of order $M$, obtained from the singularity removal constraints, using exactly $M$ symbols is advantageous for the reason outlined as follows: It was shown in \cite{APT1} that for the case when QPSK signal set is used during the MA phase, there exists certain values of fade state for which a clustering with cardinality 5 needs to be used to maximize the minimum cluster distance. While the use of clusterings with cardinality 5, reduces the impact of multiple access interference, it adversely impacts the performance during the BC phase. The 2-plexes and 4-plexes can be completed using $M$ symbols, means that the optimizing the performance during the MA phase does not come at the cost of a degraded performance during the BC phase.

 Whether a partially filled $M \times M$ Latin Square is completable using $M$ symbols or not is an open problem. Burton \cite{Bu} conjectured the following regarding the completability of the partially filled Latin Square which forms a $t$-plex of order $M$.
\begin{conjecture}
\label{Thm_comp}
A $t$-plex of a $M \times M$ Latin Square is completable using $M$ symbols, for $t \leq M/4$.  
\end{conjecture}

By explicitly providing all the Latin Squares, for 8-PSK signal set, it was shown in the previous section that the CPLS corresponding to all the singular fade states are completable using 8 symbols. If Conjecture \ref{Thm_comp} were indeed true, the 4-plexes and the 2-plexes obtained from the CPLS corresponding to the different singular fade states, can be completed using $M$ symbols, for $M$ any power of 2 greater than or equal to 16. Even though, the problem of completability of a $t$-plex in general remains unsolved, in the next section, it is shown that the structure of the singularity constraints allows the construction of explicit Latin Squares, for some singular fade states. In other words, by providing some explicit constructions, it is shown that the $4$-plexes of order $M$ corresponding to some of the singular fade states are always completable using $M$ symbols. 
     
\section{SOME SPECIAL CONSTRUCTIONS OF LATIN SQUARES}

As discussed in section II, the singular fade states lie on circles with radii $\frac{\sin\left(\frac{\pi k}{M}\right)} {\sin\left(\frac{\pi l}{M}\right)}, 1 \leq k,l \leq M/2$. It was shown in Section II that all the Latin Squares which remove the singular fade states on the unit circle (i.e., the case when $k=l$) can be obtained from the bit-wise XOR map, by appropriate column permutation. In Subsection A of this section, an explicit construction algorithm is provided to obtain Latin Squares for the case when both $k$ and $l$ are odd. In Subsection B of this section, a construction procedure to obtain Latin Squares for the case when both $k$ and $l$ are even is provided. Note that the cells in the Latin Squares obtained are filled using exactly $M$ symbols.

It was observed in \cite{APT1} that for 4 PSK signal set there exists clusterings which remove more than one singular fade state. The Latin Squares obtained from the construction algorithm given in Subsection A remove multiple singular fade states. 

The construction procedure given in Subsection B involves obtaining Latin squares for $M$-PSK signal set from the ones already obtained for $M/2$-PSK signal set.
\subsection{Construction of Latin Squares which remove Multiple Singular Fade States}
Let $\mathcal{L}_{k,l}^e$ denote the Latin Square which removes multiple singular fade states on the circle $\frac{\sin\left(\frac{\pi k}{M}\right)} {\sin\left(\frac{\pi l}{M}\right)}, 1 \leq k,l \leq M/2$, both $k$ and $l$ are odd, which is obtained as follows:

\begin{algorithmic}[1]
\STATE Start with an empty Latin Square.
\STATE The first row of the $M \times M$ Latin Square is filled with numbers $0$ to $M-1$ in the same order.
\FOR{$0 \leq s \leq M-1$}
\FOR{$1 \leq t \leq M$}
\IF{$s$ is even}
\STATE The entry $(kt,s+lt)$ is filled in the Latin Square is filled to be $s$.
\ELSE
\STATE The entry $(kt,s-lt)$ is filled in the Latin Square is filled to be $s$ (Note that all the arithmetic operations are modulo $M$.)
\ENDIF
\ENDFOR
\ENDFOR
\end{algorithmic}
Since both $k$ and $l$ are odd, the algorithm given above is always guaranteed to result in a completely filled Latin Square. 

Let $\mathcal{L}_{k,l}^o$ denote the Latin Square obtained by replacing the word \textit{even} by the word \textit{odd} in Step 5 of the algorithm.

\begin{figure}
\centering
\subfigure[]{
\includegraphics[totalheight=2in,width=2in]{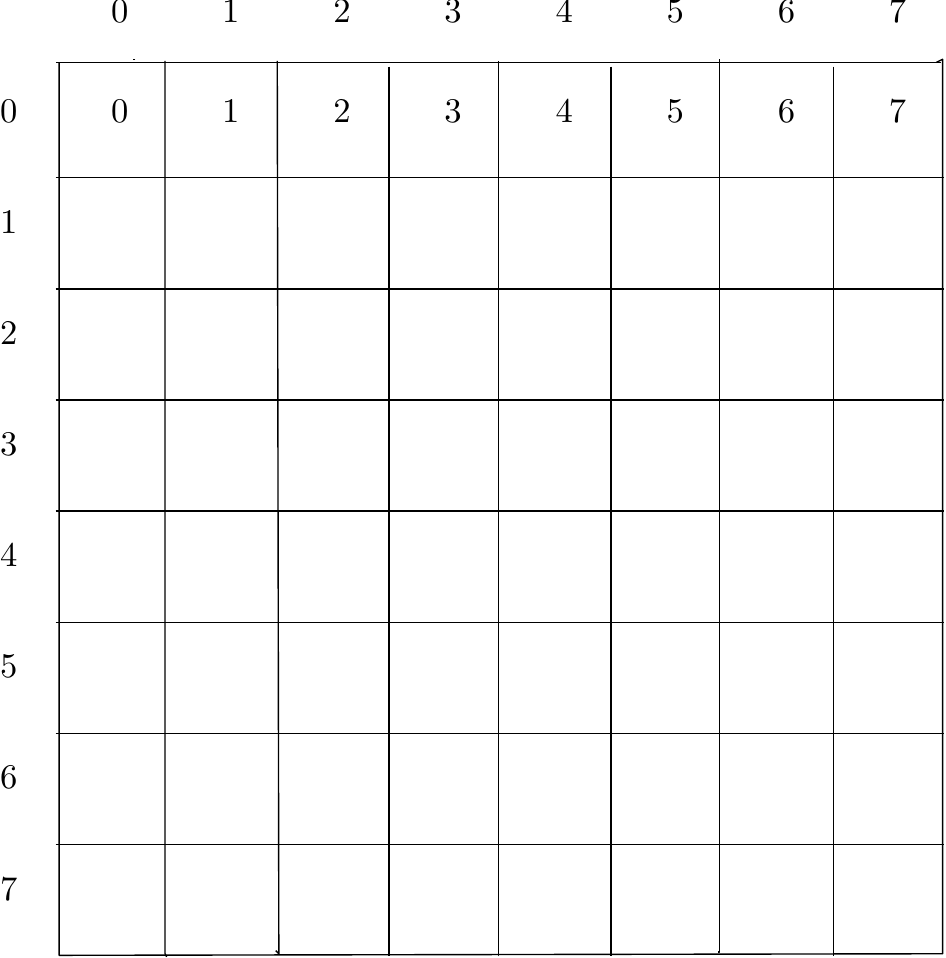}
\label{fig:latin_8psk_1}	
}

\subfigure[]{
\includegraphics[totalheight=2in,width=2in]{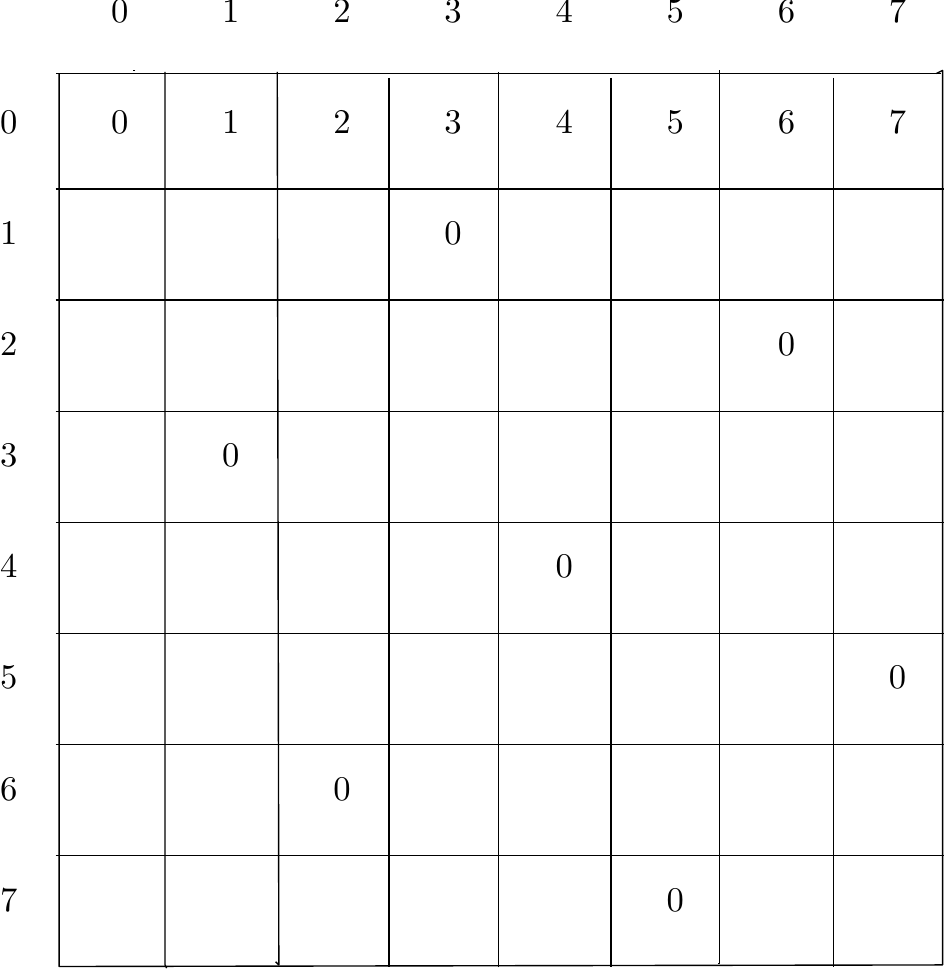}
\label{fig:latin_8psk_2}	
}

\subfigure[]{
\includegraphics[totalheight=2in,width=2in]{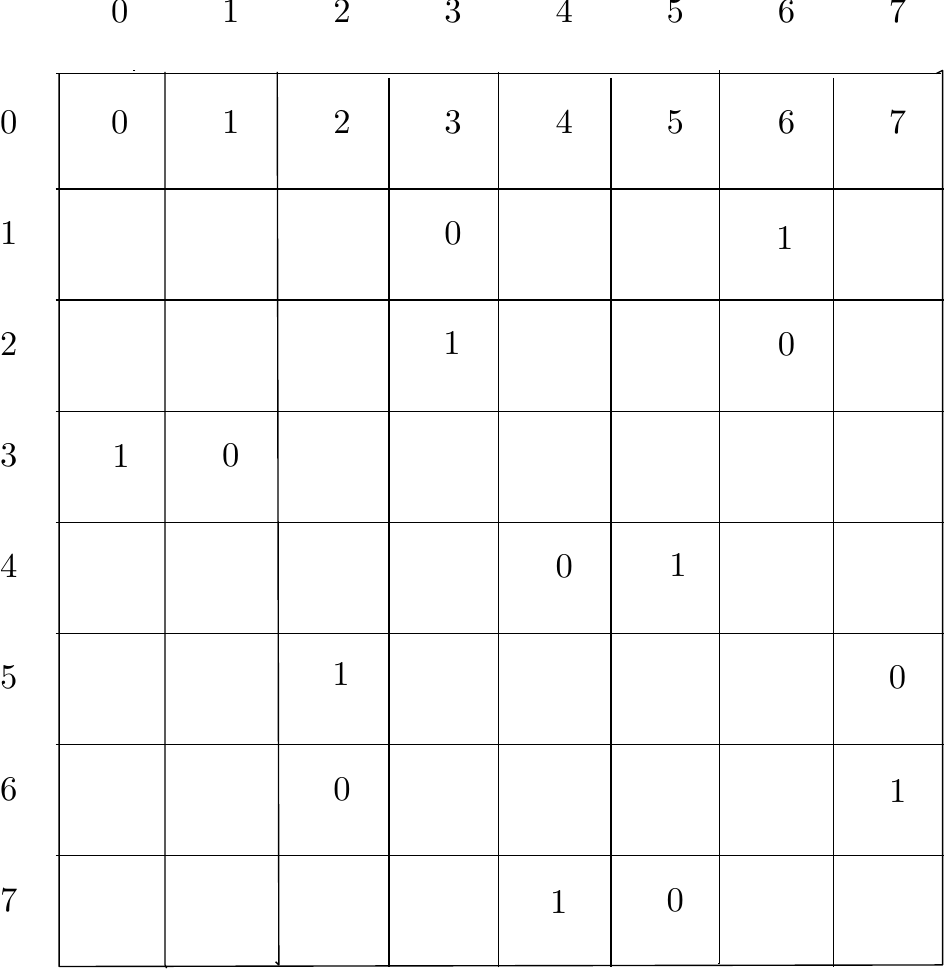}
\label{fig:latin_8psk_3}	
}

\subfigure[]{
\includegraphics[totalheight=2in,width=2in]{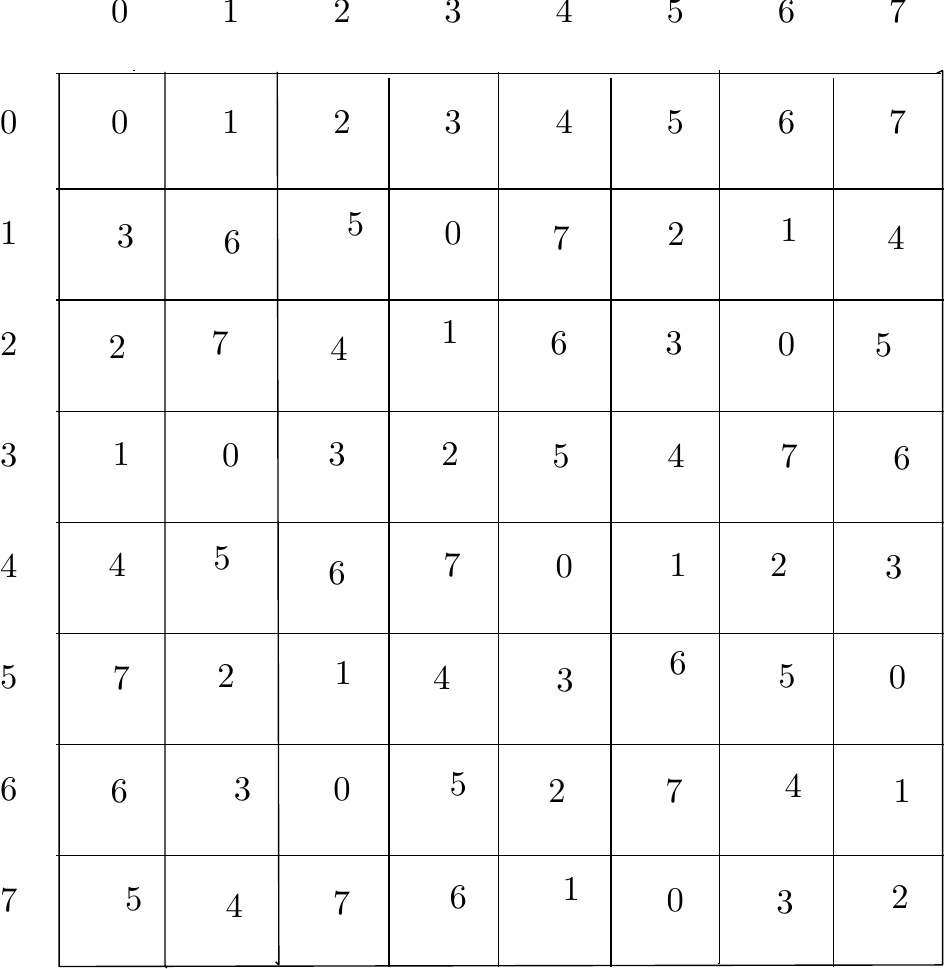}
\label{fig:latin_8psk_4}	
}
\caption{Construction of the Latin Square $\mathcal{L}_{3,1}^e$}
\label{fig:latin_8psk_even}
\end{figure}

\begin{figure}
\centering
\includegraphics[totalheight=2in,width=2in]{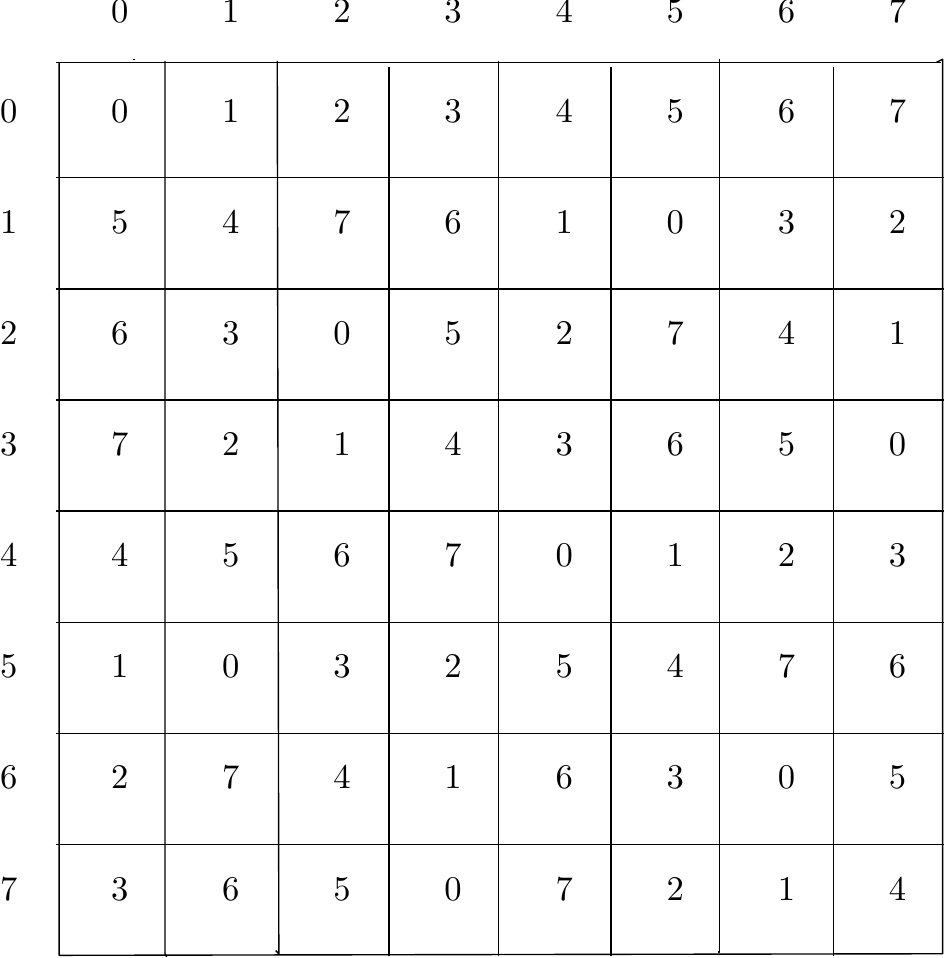}
\caption{The Latin Square $\mathcal{L}_{3,1}^o$}
\label{fig:latin_8psk_odd}
\end{figure}

For example consider the case when $k=3,l=1$. We start with only the first row filled as shown in Fig. \ref{fig:latin_8psk_1}. The cells in which $0$'s appear are filled next, by traversing the Latin Square, $k=3$ steps forward and $l=1$ step to the right, as depicted in Fig. \ref{fig:latin_8psk_2}. The cells in which $1$'s appear are filled next, by traversing the Latin Square, $k=3$ steps forward and $l=1$ step to the left, as depicted in Fig. \ref{fig:latin_8psk_3}. In a similar way, the remaining cells in the Latin Square can also be filled to obtain the Latin Square $\mathcal{L}_{3,1}^e$ shown in Fig. \ref{fig:latin_8psk_4}. Similarly, the Latin Square $\mathcal{L}_{3,1}^o$ can be obtained as shown in Fig. \ref{fig:latin_8psk_odd}.

The following Theorem states that multiple singular fade states on multiple circles are removed by the Latin Squares $\mathcal{L}_{k,l}^e$ and $\mathcal{L}_{k,l}^o$.

\begin{theorem}
\label{thm_can_LS}
One of two the Latin Squares $\mathcal{L}_{k,l}^e$ and $\mathcal{L}_{k,l}^o$ (both $k$ and $l$ odd), removes any singular fade state of the form $\frac{\sin{\frac{k_1 \pi}{M}}} {\sin{\frac{l_1 \pi}{M}}} e^{\frac{\j m 2 \pi}{M}}$, where $0 \leq m \leq M-1$ and $(k_1,l_1)$ belongs to the set, $\left\lbrace (nk,nl) \vert 1 \leq n \leq \frac{M}{2}-1, n \textrm{ odd}\right\rbrace$. 
\begin{proof}
The singularity constraints for the singular fade state $\frac{\sin{\frac{k \pi}{M}}} {\sin{\frac{l \pi}{M}}} e^{\frac{\j m 2 \pi}{M}}$ are given by Lemma \ref{lemma_sing_const1}.

 Assuming $i'=i-k$, The singularity constraints given in Lemma \ref{lemma_sing_const1} can be rewritten as,
 
 {\footnotesize
\begin{align}
\label{eqn_sing_const3}
 \left\lbrace\left(i',i'+m+\frac{M}{2}+\frac{\left(k-l\right)}{2}\right),\left(i'+k,i'+m-\frac{M}{2}+\frac{\left(k+l\right)}{2}\right)\right\rbrace\\
 \label{eqn_sing_const4}
 \left\lbrace\left(i',i'+m+\frac{\left(k+l\right)}{2}\right),\left(i'+k,i'+m+\frac{\left(k-l\right)}{2}\right)\right\rbrace,
 \end{align}
 }where $0 \leq i' \leq M-1$. 
 
From \eqref{eqn_sing_const3}, the second entry  in the the singularity constraint can be obtained from the first entry by traversing $k$ rows downwards and $l$ columns to the left. Similarly, from \eqref{eqn_sing_const4}, the second entry  in the the singularity constraint can be obtained from the first entry by traversing $k$ rows downwards and $l$ columns to the right. Furthermore, for a given $i'$, the column indices of the first cells given in the singularity constraint given in \eqref{eqn_sing_const3} and \eqref{eqn_sing_const4} differ by an odd number. Hence, one of the two Latin squares $\mathcal{L}_{k,l}^e$ and $\mathcal{L}_{k,l}^o$ removes the singular fade state $\frac{\sin{\frac{k \pi}{M}}} {\sin{\frac{l \pi}{M}}} e^{\frac{\j m 2 \pi}{M}}$, depending on whether $m+\frac{\left(k+l\right)}{2}$ is odd or even. From, the construction procedure described in the beginning of the sub-section, it can be verified easily that $\mathcal{L}_{k,l}^e$ is the same as $\mathcal{L}_{nk,nl}^e$ and $\mathcal{L}_{k,l}^o$ is the same as $\mathcal{L}_{nk,nl}^o$, for all odd $n$. This completes the proof.
\end{proof}
\end{theorem}

It follows from Theorem \ref{thm_can_LS} that a single Latin Square $\mathcal{L}_{k,l}^e$ removes half of the singular fade states which lie on $M/4$ circles, i.e., a total of $\frac{M^2}{8}$ singular fade states. Two Latin Squares $\mathcal{L}_{k,l}^e$ and $\mathcal{L}_{k,l}^o$ remove all the fade states which lie on $M/4$ circles. Also, the Latin square $\mathcal{L}_{k,l}^e$ ($\mathcal{L}_{k,l}^o$) is the same as the Latin Squares $\mathcal{L}_{nk,nl}^e$ ($\mathcal{L}_{nk,nl}^o$), where $3\leq n \leq M/2-1$, $n$ odd. 

For example, the Latin Squares $\mathcal{L}_{3,1}^e$ and $\mathcal{L}_{3,1}^o$ , shown respectively in Fig. \ref{fig:latin_8psk_4} and Fig. \ref{fig:latin_8psk_odd}, remove all the 16 singular fade states which lie on the 2 circles with radii $\frac{\sin(\frac{k_1 \pi}{M})}{{\sin(\frac{l_1 \pi}{M})}}$, where $(k_1,l_1) \in \lbrace(1,3),(3,1)\rbrace$.

\begin{note}
The Latin Square $\mathcal{L}_{k,l}^o$ is the same as the one obtained by a single column permutation of $\mathcal{L}_{k,l}^e$, with the symbols alone permuted. Since the clustering is unchanged by symbol permutation, it is enough if we specify either $\mathcal{L}_{k,l}^e$ or $\mathcal{L}_{k,l}^o$.
\end{note}

Out of the $\frac{M^2}{4}-\frac{M}{2}+1$ circles on which the singular fade states lie, $\frac{M^2}{16}-\frac{M}{4}$ circles have radii of the form  $\frac{\sin(\frac{k \pi}{M})}{{\sin(\frac{l \pi}{M})}}$, where both $k$ and $l$ are odd. From Theorem 1, the Latin Squares which remove all the singular fade states which lie on such circles can be obtained from $\frac{\frac{M^2}{16}-\frac{M}{4}}{\frac{M}{4}}=\frac{M}{4}-1$ Latin Squares, which are constructed using the algorithm given in the beginning of this subsection. 

\subsection{Construction of Latin Squares for $M$ PSK signal set from that of $M/2$ PSK signal set}

In the previous subsection, an explicit construction of Latin Squares which remove the singular fade states on the circle $\frac{\sin{\frac{\pi k }{M}}}{\sin{\frac{\pi l }{M}}}$, both $k$ and $l$ odd, was provided. In this subsection, it is shown that for $M$ PSK signal set, Latin Squares which remove the singular fade states on the circle $\frac{\sin{\frac{\pi k }{M}}}{\sin{\frac{\pi l }{M}}}$, both $k$ and $l$ even, $k \neq l, k \neq M/2$ and $l \neq M/2$, can be constructed from the Latin Squares obtained for $M/2$ PSK signal set.

Let $\mathcal{L}(M,k,l,\phi)$ denote an $M \times M$ Latin Square which removes the singular fade state $\frac{\sin(\frac{k\pi}{M})} {\sin(\frac{l\pi}{M})} e^{j \phi}$. 

For the $M \times M$ Latin Square $\mathcal{L}(M,k,l,\phi)$, let $\mathcal{L}^{oe}(M,k,l,\phi)$ denote the $M/2 \times M/2$ Latin Square obtained by taking only the odd numbered rows and even numbered columns of $\mathcal{L}(M,k,l,\phi)$. In a similar way, the  the $M/2 \times M/2$ Latin Squares $\mathcal{L}^{ee}(M,k,l,\phi),\mathcal{L}^{oo}(M,k,l,\phi),\mathcal{L}^{eo}(M,k,l,\phi)$ can be defined.

\begin{definition}
The \textit{quadruplicate} of an $M \times M$ Latin Square $\mathcal{L}(M,k,l,\phi)$ is the ordered set consisting of the four $M/2 \times M/2$ Latin Squares 
{\small
$\lbrace \mathcal{L}^{ee}(M,k,l,\phi),\mathcal{L}^{oo}(M,k,l,\phi),\mathcal{L}^{eo}(M,k,l,\phi),\mathcal{L}^{oe}(M,k,l,\phi) \rbrace.$
}  
\end{definition} 

The Latin Square is $\mathcal{L}(M,k,l,\phi)$ uniquely determined by its quadruplicate. 

For the Latin Square $\mathcal{L}(M,k,l,\phi)$, let $\mathcal{L}_{c}(M,k,l,\phi)$ denote the Latin Square obtained by adding integer $c$ to all the cells of the Latin Square.

The following theorem gives the construction of the Latin Square $\mathcal{L}(M,k,l,\phi)$, for the case when both $k$ and $l$ are even, $k \neq l, k \neq M/2$, $l \neq M/2$ and $\mod\left(\frac{k}{2}+\frac{l}{2},2\right)=0$.
\begin{theorem}
\label{theorem_even_odd1}
The elements of the quadruplicate of $\mathcal{L}(M,k,l,\phi)$, if $\mod\left(\frac{k}{2}+\frac{l}{2},2\right)=0$, are given by,

{\footnotesize
\begin{align*}
\mathcal{L}^{ee}(M,k,l,\phi)=\mathcal{L}^{oo}(M,k,l,\phi)&=\mathcal{L}(M/2,k/2,l/2,\phi),\\
 \mathcal{L}^{eo}(M,k,l,\phi)=\mathcal{L}^{oe}(M,k,l,\phi)&=\mathcal{L}_{M/2}(M/2,k/2,l/2,\phi).
 \end{align*} }
\begin{proof}
From Lemma \ref{lemma_sing_const1}, the singularity constraints are given by,

 {\footnotesize
\begin{align}
\label{eqn_sing_const5}
 \left\lbrace\left(i,i-\frac{M}{2}-\frac{\left(k-l\right)}{2}\right),\left(i-k,i+\frac{M}{2}-\frac{\left(k+l\right)}{2}\right)\right\rbrace\\
 \label{eqn_sing_const6}
 \left\lbrace\left(i,i-\frac{\left(k+l\right)}{2}\right),\left(i-k,i-\frac{\left(k-l\right)}{2}\right)\right\rbrace,
 \end{align}
 }where $0 \leq i \leq M-1$. 
 
 The singularity constraints given above can be further split into cases depending on whether $i$ is even or odd. 
 
\noindent{Case 1: $i$ is even}
 \noindent 
 
 Assume $i'=i/2$, the singularity constraints in \eqref{eqn_sing_const5} and \eqref{eqn_sing_const6} can be rewritten as,

 {\footnotesize
\begin{align}
\nonumber
 &\left\lbrace\left(2i',2i'-2\frac{M/2}{2}-\frac{2\left(k/2-l/2\right)}{2}\right)\right.,\\
\label{eqn_sing_const7}
&\hspace{2 cm}\left.\left(2i'-2\frac{k}{2},2i'+2\frac{M/2}{2}-2\frac{\left(k/2+l/2\right)}{2}\right)\right\rbrace\\ 
\nonumber
&\left\lbrace\left(2i',2i'-2\frac{\left(k/2+l/2\right)}{2}\right),\right.\\
 \label{eqn_sing_const8}
 &\hspace{3 cm}\left.\left(2i'-2\frac{k}{2},2i'-2\frac{\left(k/2-l/2\right)}{2}\right)\right\rbrace,
 \end{align}
 }where $0 \leq i' \leq M/2-1$. 
 
 The singularity constraints given in \eqref{eqn_sing_const7} and \eqref{eqn_sing_const8}, are the same as that of the singular fade state $\frac{\sin(\frac{\pi k/2}{M/2})}{\sin(\frac{\pi l/2}{M/2})}$ with the row and column indices doubled. Hence, $\mathcal{L}(M,2k,2l,\phi)^{ee}=\mathcal{L}(M/2,k,l,\phi)$.
 
 \noindent{Case 2: $i$ is odd}
 \noindent 
 
  Using a similar argument as that of case 1, it can be shown that the singularity constraints for this case are the same as that of the singular fade state $\frac{\sin(\frac{\pi k/2}{M/2})}{\sin(\frac{\pi l/2}{M/2})}$ with the row and column indices doubled and shifted by one place. Hence, $\mathcal{L}^{oo}(M,k,l,\phi)=\mathcal{L}(M/2,k/2,l/2,\phi)$. Note that by taking both $\mathcal{L}^{ee}(M,k,l,\phi)$ and $\mathcal{L}^{oo}(M,k,l,\phi)$ to be $\mathcal{L}(M/2,k/2,l/2,\phi)$, the conditions required for $\mathcal{L}(M,k,l,\phi)$ to be a Latin Square are not violated.
 
 Since  {\footnotesize $\mathcal{L}^{ee}(M,k,l,\phi)=\mathcal{L}^{oo}(M,k,l,\phi)=\mathcal{L}(M/2,k/2,l/2,\phi)$}, half of the cells in the Latin Square $\mathcal{L}(M,k,l,\phi)$ are filled with numbers from $0$ to $M/2-1$. Hence, $\mathcal{L}^{oe}(M,k,l,\phi)$ and $\mathcal{L}^{eo}(M,k,l,\phi)$ can be taken to be $\mathcal{L}_{M/2}(M/2,k/2,l/2,\phi)$.
 
\end{proof}
\end{theorem}

\begin{figure}
\centering
\subfigure[]{
\includegraphics[totalheight=2in,width=2in]{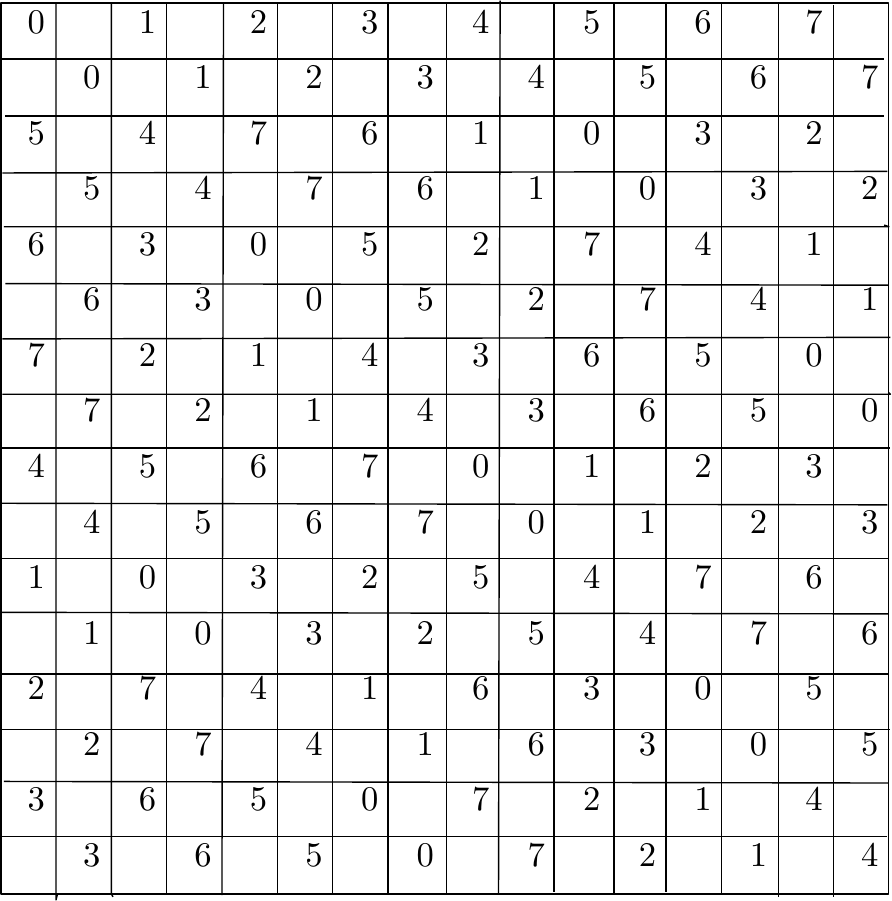}
\label{fig:latin_16psk_2}	
}

\subfigure[]{
\includegraphics[totalheight=2in,width=2in]{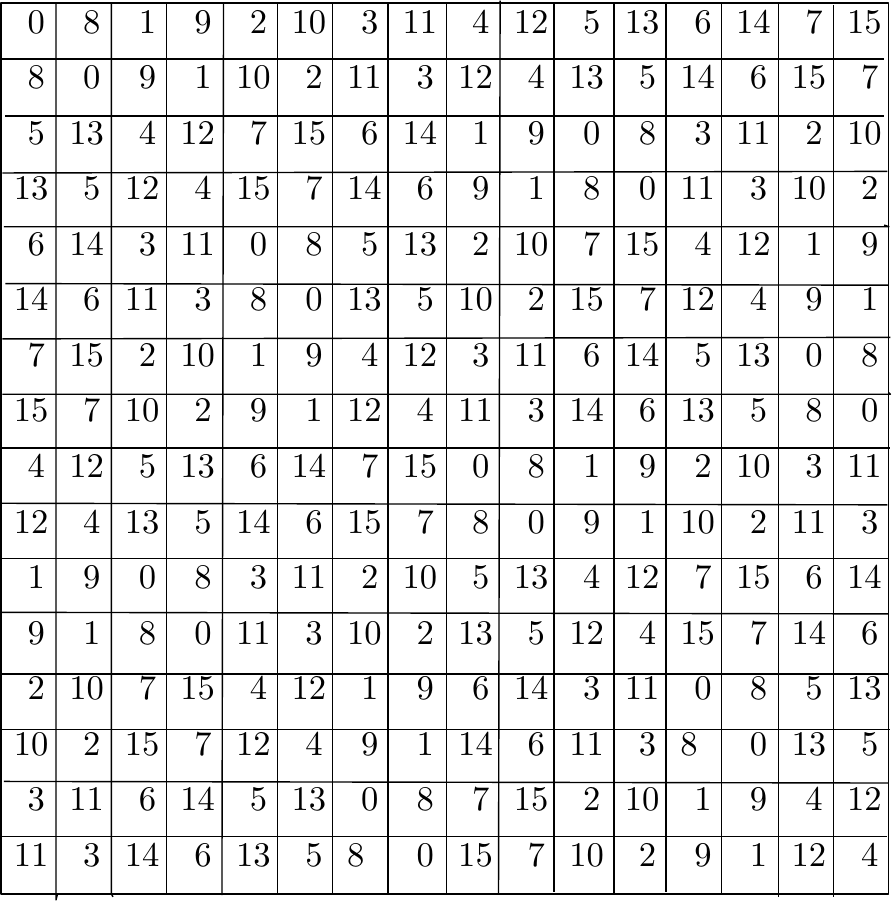}
\label{fig:latin_16psk_3}	
}
\caption{Construction of the Latin Square $\mathcal{L}(16,6,2,0)$}
\label{fig:latin_16psk_even}
\end{figure}

For example, the Latin Square $\mathcal{L}(16,6,2,0)$ can be obtained from $\mathcal{L}(8,3,1,0)$ (which is shown in Fig. \ref{fig:latin_8psk_odd}) as follows: $\mathcal{L}^{ee}(16,6,2,0)$ and $\mathcal{LS}^{oo}(16,6,2,0)$ are chosen to be $\mathcal{L}(8,3,1,0)$ as shown in Fig. \ref{fig:latin_16psk_2}. The Latin Squares $\mathcal{L}^{oe}(16,6,2,0)$ and $\mathcal{L}^{eo}(16,6,2,0)$ are chosen to be $\mathcal{L}_8(8,3,1,0)$ to obtain the Latin Square $\mathcal{L}(16,6,2,0)$ shown in Fig. \ref{fig:latin_16psk_3}.

The following theorem gives the construction of the Latin Square $\mathcal{L}(M,k,l,\phi)$, for the case when both $k$ and $l$ are even, $k \neq l, k \neq M/2$, $l \neq M/2$ and $\mod\left(\frac{k}{2}+\frac{l}{2},2\right)=1$.
\begin{theorem}
The elements of the quadruplicate of $\mathcal{L}(M,k,l,\phi)$, if $\mod\left(\frac{k}{2}+\frac{l}{2},2\right)=1$, are given by,

{\footnotesize
\begin{align*}
&\mathcal{L}^{oe}(M,k,l,\phi)=\mathcal{L}(M/2,k/2,l/2,\phi),\\
&\mathcal{L}^{eo}(M,k,l,\phi)=\mathcal{L}(M/2,k/2,l/2,\phi-2\pi/M),\\
&\mathcal{L}^{oo}(M,k,l,\phi)=\mathcal{L}^{ee}(M,k,l,\phi)=\mathcal{L}_{M/2}(M/2,k/2,l/2,\phi).
 \end{align*} }
\begin{proof}
The proof is similar to Theorem \ref{theorem_even_odd1} and is omitted.
\end{proof}
\end{theorem}

In this subsection, construction procedure was provided to obtain the Latin Squares which remove the singular fade states which lie on circles with radii $\frac{\sin(\frac{k \pi}{M})}{\sin(\frac{l \pi}{M})}$, where both $k$ and $l$ are even, $k \neq l, k \neq M/2$ and $l \neq M/2$. The total number of such circles is $\left(\frac{M}{4}-1\right)^2-\left(\frac{M}{4}-1\right)=\frac{M^2}{16}-\frac{3M}{4}+2$. It is enough if we obtain one Latin Square per circle for those circles which lie inside the unit circle, since the Latin Square corresponding to all other singular fade states can be obtained by appropriately performing column permutation and taking transpose. Hence, to remove all the singular fade states on circles with radii $\frac{\sin(\frac{k \pi}{M})}{\sin(\frac{l \pi}{M})}$, where both $k$ and $l$ are even, $k \neq l, k \neq M/2$ and $l \neq M/2$, it is enough if we obtain $\frac{\left(\frac{M^2}{16}-\frac{3M}{4}+2\right)}{2}$ Latin Squares from the construction procedure described in this subsection.

The Latin Squares constructed so far (the bit-wise XOR map and the ones obtained in Subsection A and B of this section), which are 

{\footnotesize 
$$1+\left(\frac{M}{4}-1\right)+\frac{\left(\frac{M^2}{16}-\frac{3M}{4}+2\right)}{2}=\frac{M^2}{32}-\frac{M}{8}+1$$
}in number, are sufficient to obtain all the Latin Squares which remove the singular fade states which lie on 

{\footnotesize
$$1+\left(\frac{M^2}{16}-\frac{M}{4}\right)+\left(\frac{M^2}{16}-\frac{3M}{4}+2\right)=\frac{M^2}{8}-M+3$$
}circles. The Latin Squares which remove the singular fade states which lie on the other circles inside the unit circle can be found using computer search. It can be verified that it is sufficient to obtain $\frac{3M^2}{32}+\frac{M}{8}$ Latin Squares in total, from which all the other Latin Squares can be obtained, for $M \geq 8$.
\section{Bidirectional Relaying With Unequal Transmission Rates and Latin Rectangles}

In this section we consider the scenario in which both the end nodes use PSK constellations of different sizes. This results in different transmission rates for both the users. Assume node A and node B use PSK constellations $\mathcal{S}_A$ and $\mathcal{S}_B$ of size $M (=2^\lambda)$ and $N (=2^\mu)$ respectively. The relay should use a constellation of size at least $\max(M,N)$. The exclusive law for this case is given by,

{\footnotesize
\begin{align}
\left.
\begin{array}{ll}
\nonumber
\mathcal{M}^{\gamma,\theta}(x_A,x_B) \neq \mathcal{M}^{\gamma,\theta}(x'_A,x_B), \; \mathrm{where} \;x_A \neq x'_A \; \mathrm{,} \;x_B \in  \mathcal{S}_B,\\
\nonumber
\mathcal{M}^{\gamma,\theta}(x_A,x_B) \neq \mathcal{M}^{\gamma,\theta}(x_A,x'_B), \; \mathrm{where} \;x_B \neq x'_B \; \mathrm{,} \;x_A \in \mathcal{S}_A.
\end {array}
\right\} \\
\label{ex_law}
\end{align}
}
\begin{definition} A Latin rectangle L on the symbols from the set $\mathbb{Z}_t=\{0,1, \cdots ,t-1\}$ is an \textit{M} $\times$ \textit{N}  array, in which each cell contains one symbol and each symbol occurs at most once in each row and column.  
\end{definition}

 The relay clusterings which satisfy the exclusive law form Latin rectangles with $M$ rows and $N$ columns. As before, the study of clusterings which satisfy the exclusive law can be equivalently carried out as the study of Latin rectangles with appropriate parameters. Without loss of generality, it is assumed that $M> N$. It is shown that the set of singular fade states for this case is a subset of the set of singular fade states for the case when both A and B use $M$-PSK signal set. Also, it is shown that the $M \times N$ Latin Rectangle for this case can be obtained from $M \times M$ Latin Square obtained for the case when both A and B use $M$-PSK signal set.
 
\begin{lemma}
The set of singular fade states for the case when $M$-PSK signal set is used at A and $N$-PSK signal set is used at B, both $M$ and $N$ powers of 2 and $M >N$, is a subset of the set of singular fade states for the case when both A and B use $M$-PSK signal set.
\end{lemma}
\begin{proof}
Let $\mathcal{S}_M$ and $\mathcal{S}_N$ denote the symmetric $M$-PSK and $N$-PSK signal sets used at A and B respectively. 

Assume $M=2^{\lambda}$ and $N=2^{\mu}$, $\lambda > \mu$.

Note that the signal set $\mathcal{S}_N$ is a subset of the signal set $\mathcal{S}_M$.

The singular fade states for this case are of the form,

{\footnotesize
\begin{align}
\label{eqn_sing_rectan}
\gamma e^{j \theta}=-\frac{x(k)-x(k')}{y(l)-y(l')}, x(k)\neq x(k') \in \mathcal{S}_M, y(l)\neq y(l') \in \mathcal{S}_N
\end{align}
}

The singular fade states given by \eqref{eqn_sing_rectan}, are also singular fade states for the case when both A and B use $M$-PSK signal set, since $y(l)$ and $y(l')$ are points in the $M$-PSK signal set as well. This completes the proof.
\end{proof}

 The $M \times N$ system can be interpreted as a $M \times M$ system with node B not transmitting some symbols. The corresponding Latin Rectangle is obtained from the Latin Square for $M \times M$ system by removing the columns corresponding to the unused symbols.
 The following example illustrates this.
\begin{example}
  Consider the case when 8-PSK signal set is used at A and 4-PSK signal set is used at B. For $\gamma=\sin(\pi/4)$ and $\theta=0$, by removing the columns  1, 3, 5 and 7 (equivalently viewed as B transmitting only the 4 symbols 0, 2, 4 and 6 in the 8-PSK signal set) from the Latin Square given in Fig. \ref{8L6}, the Latin Rectangle shown in Fig. \ref{latin_rectang} is obtained.
\end{example}

\begin{figure}[htbp]
\centering
{
\begin{tabular}{|c|c|c|c|}
\hline 4 & 6 & 1 & 7 \\ 
\hline 2 & 0 & 3 & 5 \\
\hline 0 & 7 & 5 & 6 \\  
\hline 3 & 4 & 2 & 1 \\ 
\hline 5 & 2 & 0 & 3 \\ 
\hline 7 & 1 & 6 & 4 \\
\hline 1 & 3 & 4 & 2 \\ 
\hline 6 & 5 & 7 & 0 \\
\hline
\end{tabular}}
\caption{Latin Rectangle for  $\gamma=\sin(\pi/4), \theta=0$}
\label{latin_rectang}
\end{figure}
\vspace{-0.1 cm}
\section{DISCUSSION}
In this paper, for the design of modulation schemes for the physical layer network-coded two way relaying scenario with the protocol which employs two phases: Multiple access (MA) Phase and Broadcast (BC) phase, we identified a relation between the required exclusive laws satisfying clusterings and Latin Squares. This relation is  used to get all the maps to be used at the relay efficiently. Further we illustrated the results presented for the case, where both the end nodes use QPSK constellation as well as 8-PSK constellations. Here we concentrated only on singular fade states and the clusterings to remove that with only the minimum cluster distance under consideration. We are not considering the entire distance profile as done in \cite{APT1}. Our work eliminate the singular fade states effectively and these clusterings can be used in other regions in the complex plane of $(\gamma, \theta)$, as shown in  \cite{VNR}.
\section*{Acknowledgement}
This work was supported  partly by the DRDO-IISc program on Advanced Research in Mathematical Engineering through a research grant as well as the INAE Chair Professorship grant to B.~S.~Rajan.
\end{document}